\documentclass[a4paper,reqno]{amsart}
\usepackage{amsmath,amssymb,mathrsfs}
\usepackage{txfonts,bm,pifont}
\usepackage{cite}
\usepackage{float}
\usepackage{stackrel}
\usepackage{subcaption}
\usepackage{graphicx,xcolor}
\usepackage{comment}
\usepackage{enumitem}
\usepackage{longtable}
\usepackage{hhline}
\usepackage {diagbox}
\usepackage{hyperref}

\newtheorem{theorem}{Theorem}[section]

\newtheorem{remark}[theorem]{Remark}
\theoremstyle{definition}

\numberwithin{equation}{section}
\numberwithin{figure}{section}
\numberwithin{table}{section}
\newcommand{\wutilde}[1]{\vrule depth 0pt width 0pt%
{\raise0.8pt\hbox{$\smash{{\mathop{#1} \limits_{\displaystyle\widetilde{}}}}$}}}
\newcommand{\wuhat}[1]{\vrule depth 0pt width 0pt%
{\raise0.6pt\hbox{$\smash{{\mathop{#1} \limits_{\displaystyle\widehat{}}}}$}}}

\newcommand{\al}{\alpha}
\newcommand{\be}{\beta}
\newcommand{\de}{\delta}
\newcommand{\ga}{\gamma}

\newcommand{\si}{\sigma}

\newcommand{\la}{\lambda}

\newcommand{\ep}{\epsilon}
\newcommand{\ka}{\kappa}
\newcommand{\PDE}{P$\Delta$E}
\newcommand{\ODE}{O$\Delta$E}
\newcommand{\bbZ}{\mathbb{Z}}

\newcommand{\bbC}{\mathbb{C}}

\newcommand{\set}[2]{\left\{\left. #1 ~\right|~ #2 \right\}}

\makeatletter
\long\def\@makecaption#1#2{
 \vskip 10pt
 \setbox\@tempboxa\hbox{#1. #2}
 \ifdim \wd\@tempboxa >\hsize #1. #2\par \else \hbox
to\hsize{\hfil\box\@tempboxa\hfil}
 \fi}
\makeatother
\makeatletter
\@namedef{subjclassname@2020}{%
  \textup{2020} Mathematics Subject Classification}
\makeatother
\begin{document}
\allowdisplaybreaks

\title[]{Special solutions to five autonomous integrable partial difference equations via the third and sixth Painlev\'e equations and the Garnier system in two variables}
\author{Nobutaka Nakazono}
\address{Institute of Engineering, Tokyo University of Agriculture and Technology, 2-24-16 Nakacho Koganei, Tokyo 184-8588, Japan.}
\email{nakazono@go.tuat.ac.jp}
\begin{abstract}
In this paper, we study special solutions of five autonomous integrable partial difference equations ({\PDE}s).
More precisely, we show that these {\PDE}s admit special solutions that are described by non-autonomous ordinary difference equations arising from B\"acklund transformations of the third and sixth Painlev\'e equations and the Garnier system in two variables.
This result provides a new perspective on the relationship between autonomous integrable {\PDE}s and Painlev\'e-type dynamics.
\end{abstract}

\maketitle

\section{Introduction}\label{Introduction}
Integrable partial difference equations ({\PDE}s) on lattices play an important role in mathematical physics, appearing in soliton theory and discrete differential geometry.
Among them, integrable two-dimensional {\PDE}s admitting consistency properties form a fundamental class.
See, {\it e.g.}, \cite{HJN2016:MR3587455,BS2008:zbMATH05486618} and references therein.

It is natural that special solutions of non-autonomous integrable {\PDE}s are described in terms of non-autonomous integrable ordinary difference equations ({\ODE}s) (see, {\it e.g.}, \cite{NP1991:MR1098879,GRSWC2005:MR2117991,FJN2008:MR2425981,HHJN2007:MR2303490,OrmerodCM2012:MR2997166,HHNS2015:MR3317164,OrmerodCM2014:MR3210633,nakazono2022discrete}).
By contrast, the fact that special solutions of autonomous integrable {\PDE}s can be described in terms of non-autonomous integrable {\ODE}s is far from obvious, and such phenomena have been reported only in a limited number of works (see, {\it e.g.}, \cite{JKMN2021:zbMATH07653201,nakazono2026variationdKdV,NRGO2001:MR1819383}).
The mechanisms by which non-autonomous structures emerge at the level of special solutions of autonomous {\PDE}s therefore remain not yet fully understood.

In this paper, we show that the following five autonomous two-dimensional {\PDE}s admit special solutions described in terms of the non-autonomous {\ODE}s arising from the B\"acklund transformations of the third and sixth Painlev\'e equations and the Garnier system in two variables.
Remarkably, although the equations themselves are autonomous, their special solutions are governed by non-autonomous {\ODE}s.
\begin{itemize}
\item 
Hirota's discrete Korteweg-de Vries (dKdV) equation \cite{HirotaR1977:MR0460934}:
\begin{equation}\label{eqn:auto_dkdv}
 u_{l+1,m+1}-u_{l,m}=\dfrac{~1~}{u_{l,m+1}}-\dfrac{~1~}{u_{l+1,m}}.
\end{equation}
\item 
The $Q1_{\de=1}$ equation \cite{ABS2003:MR1962121,ABS2009:MR2503862}:
\begin{equation}\label{eqn:auto_Q1}
 \dfrac{(u_{l,m}-u_{l,m+1})(u_{l+1,m}-u_{l+1,m+1})}{\be(\be-\al)}
 -\dfrac{(u_{l,m}-u_{l+1,m})(u_{l,m+1}-u_{l+1,m+1})}{\al(\be-\al)}
 =1.
\end{equation}
\item 
A {\PDE} found by Hietarinta and Viallet (corresponding to the case $r_3=-1$, $r_4=0$, and $u=0$ in Equation (19) of \cite{HV2007:Searching}):
\begin{equation}\label{eqn:auto_HV}
 (u_{l,m}-u_{l,m+1})(u_{l+1,m}-u_{l+1,m+1})=u_{l,m+1}-u_{l+1,m}.
\end{equation}
Hereafter, Equation \eqref{eqn:auto_HV} is referred to as the HV equation.
\item 
The lattice sine-Gordon (lsG) equation \cite{volkov1992quantum,bobenko1993discrete}:
\begin{equation}\label{eqn:auto_lsG}
 \dfrac{~u_{l+1,m+1}~}{u_{l,m}}
 =\left(\dfrac{\al-\be\,u_{l+1,m}}{\be-\al\,u_{l+1,m}}\right)
 \left(\dfrac{\be-\al\,u_{l,m+1}}{\al-\be\,u_{l,m+1}}\right).
\end{equation}
\item 
The discrete Volterra (dVolterra) equation \cite{NC1995:MR1329559,hirota1995conserved}:
\begin{equation}\label{eqn:auto_dV}
 \dfrac{~u_{l+1,m+1}~}{u_{l,m}}=\dfrac{1+u_{l+1,m}}{1+u_{l,m+1}}.
\end{equation}
\end{itemize}
Here, $(l, m)\in\bbZ^2$ are lattice parameters, $u_{l,m}\in \bbC$ is a dependent variable, and $\al,\be\in\bbC^\ast$ are parameters.

\begin{remark}
The $Q1_{\de=1}$ equation \eqref{eqn:auto_Q1} and the lsG equation \eqref{eqn:auto_lsG} each have exactly one essential parameter.
Indeed, by setting $u_{l,m}=\al\,U_{l,m}$ and $\ga=\al\,\be^{-1}$, the parameters of Equation \eqref{eqn:auto_Q1} reduce to one parameter $\ga$.
Similarly, by setting $\ga=\al\,\be^{-1}$, Equation \eqref{eqn:auto_lsG} also depends on only one parameter.
\end{remark}

\begin{remark}
The HV equation \eqref{eqn:auto_HV} can be obtained from the $Q1_{\de=1}$ equation \eqref{eqn:auto_Q1} by a limiting procedure.
Indeed, substituting
\begin{equation}
 \al=\ep(\ep+1),\quad
 \be=\ep^2,\quad
 u_{l,m}=\ep\left(2\ep U_{l,m}-(\ep+1)l-\ep\,m\right),
\end{equation}
into Equation \eqref{eqn:auto_Q1} and taking the limit $\ep\to 0$, one obtains Equation \eqref{eqn:auto_HV}.
\end{remark}

The construction of the special solutions in this paper does not rely on any known reduction procedure.
In a related setting, special solutions of {\PDE}s are known as {\it discrete Painlev\'e transcendent solutions}
(hereafter referred to as dP solutions) \cite{nakazono2022discrete} if
(i) the evolution of the {\PDE} along each lattice direction corresponds to the time evolution of a Painlev\'e-type {\ODE}, and
(ii) the special solutions of the {\PDE} are given as rational functions of the dependent variables of these Painlev\'e-type {\ODE}s.
The special solutions obtained in this paper can also be regarded as belonging to this class in this sense.
However, since our approach treats difference equations as B\"acklund transformations of differential Painlev\'e equations and Garnier systems, and places emphasis not only on Painlev\'e equations but also on Garnier systems, we do not explicitly use the terminology of dP solutions in what follows.
\subsection{Background}
In this subsection, we briefly describe the integrable systems considered in this paper.

\subsubsection*{\bf Equations \eqref{eqn:auto_dkdv}--\eqref{eqn:auto_dV}}
Here we give brief remarks on Equations \eqref{eqn:auto_dkdv}--\eqref{eqn:auto_dV}.
\begin{itemize}
\item
The dKdV equation \eqref{eqn:auto_dkdv} is a discrete version of the Korteweg-de Vries (KdV) equation\cite{KDV1895:zbMATH02679684}:
\begin{equation}\label{eqn:KdV}
 w_t+6ww_x+w_{xxx}=0,
\end{equation}
where $(t,x)\in\bbC^2$ and $w=w(t,x)\in\bbC$,
which is known as a mathematical model for waves on shallow water surfaces.
\item
The lsG equation \eqref{eqn:auto_lsG} is a discrete version of the sine-Gordon equation:
\begin{equation}\label{eqn:sg}
 \phi_{tt}-\phi_{xx}+\sin{\phi}=0,
\end{equation}
where $(t,x)\in\bbC^2$ and $\phi=\phi(t,x)\in\bbC$,
which is known as the equation of motion for a row of pendulums suspended from a rod and coupled by torsion springs.
\item
In \cite{NC1995:MR1329559}, it is stated that the dVolterra equation \eqref{eqn:auto_dV} is a discrete version of the Volterra-Kac-van Moerbeke equation \cite{KM1975:MR0369953}:
\begin{equation}
\begin{cases}
 \dfrac{dR_1}{dt}=-e^{-R_2},\\[0.8em]
 \dfrac{dR_n}{dt}=e^{-R_{n-1}}-e^{-R_{n+1}},\quad n\in\bbZ_{\geq2},
\end{cases}
\end{equation}
where $t\in\bbC$ and $R_k=R_k(t)\in\bbC$.
On the other hand, in \cite{hirota1995conserved}, the dVolterra equation \eqref{eqn:auto_dV} is referred to as the discrete Lotka-Volterra equation of type I, and it is shown to be a discrete analogue of the Lotka-Volterra equation \cite{hirota1976n}:
\begin{equation}
\dfrac{dv_n}{dt}=v_n(v_{n-1}+v_{n+1}),
\end{equation}
where $t\in\bbC$ and $v_k=v_k(t)\in\bbC$.
\item
A relation among the four variables $\{u_{l,m},u_{l+1,m},u_{l,m+1},u_{l+1,m+1}\}$ given by an irreducible multilinear polynomial in four variables, such as Equations \eqref{eqn:auto_dkdv}--\eqref{eqn:auto_dV}, is called a {\it quad-equation}.
As a notion of integrability for quad-equations, {\it the consistency around a cube (CAC) property} is well known
(see \cite{NQC1983:MR719638,NCWQ1984:MR763123,QNCL1984:MR761644,NS1998:zbMATH01844203,NW2001:MR1869690} for details).
There also exists another integrability notion, called {\it the consistency around a broken cube (CABC) property}, which is closely related to the CAC property (see \cite{JN2021:zbMATH07476241} for details).
For the five {\PDE}s considered in this paper, the following statements hold.
\begin{itemize}
\item
Equations \eqref{eqn:auto_dkdv}--\eqref{eqn:auto_lsG} possess the CAC property.
For the CAC property of the dKdV equation \eqref{eqn:auto_dkdv} and the lsG equation \eqref{eqn:auto_lsG}, see \cite{nakazono2024consistency};
for that of the $Q1_{\de=1}$ equation \eqref{eqn:auto_Q1}, see \cite{ABS2003:MR1962121};
and for that of the HV equation \eqref{eqn:auto_HV}, see Appendix~\ref{appendix:CAC}.
\item
Equations \eqref{eqn:auto_dkdv}, \eqref{eqn:auto_lsG}, and \eqref{eqn:auto_dV} possess the CABC property.
For the CABC property of the dKdV equation \eqref{eqn:auto_dkdv} and the lsG equation \eqref{eqn:auto_lsG}, see \cite{JN2021:zbMATH07476241,nakazono2024consistency};
for that of the dVolterra equation \eqref{eqn:auto_dV}, see Appendix~\ref{appendix:CABC}.
\end{itemize}
\end{itemize}

\subsubsection*{\bf Painlev\'e equations and Garnier systems}
In the early 20th century, to find a new class of special functions, Painlev\'e\cite{PainleveP1900:zbMATH02665472,PainleveP1902:MR1554937,PainleveP1907:zbMATH02647172}  and Gambier\cite{GambierB1910:MR1555055} classified all the ordinary differential equations of the type 
\begin{equation}
 y''=F(y',y,t),
\end{equation}
where $y=y(t)$, $'=d/dt$, and $F$ is a function meromorphic in $t\in\bbC$ and rational in $y$ and $y'$, 
with the Painlev\'e property (solutions do not have movable singularities other than poles).
As a result, they obtained six new equations that are collectively referred to as Painlev\'e equations.
The Painlev\'e equations are numbered beginning with one: P$_{\rm I}$, $\dots$, P$_{\rm VI}$,
and starting from the sixth Painlev\'e equation, we can, through appropriate limiting processes, obtain other Painlev\'e equations:
$$\begin{array}{ccccccc}
 {\rm P}_{\rm VI}~(D_4^{(1)})&\to&{\rm P}_{\rm V}~(A_3^{(1)})&\to&{\rm P}_{\rm III}~(2A_1^{(1)})&\\
 &&\downarrow&&\downarrow&&\\
 &&{\rm P}_{\rm IV}~(A_2^{(1)})&\to&{\rm P}_{\rm II}~(A_1^{(1)})&\to&{\rm P}_{\rm I}
\end{array}$$
It is well known that the Painlev\'e equations, except for P$_{\rm I}$, have B\"acklund transformations, which collectively form (extended) affine Weyl groups. 
The symbols inside the parentheses of the diagram above indicate the types of affine Weyl groups.
(See \cite{OKSO2006:MR2277519,OkamotoK1987:MR927186,OkamotoK1987:MR914314,OkamotoK1987:MR916698,OkamotoK1986:MR854008,book_NoumiM2004:MR2044201,SakaiH2001:MR1882403,KNY2017:MR3609039} for the details.)
Note that P$_{\rm VI}$ was found by Fuchs \cite{FuchsR1905:quelques} through his study of linear differential equations, prior to the work of Painlev\'e {\it et al.}

Garnier systems are multivariable generalizations of P$_{\rm VI}$, which were derived by Garnier \cite{GarnierR1912:MR1509146} using the method developed by Fuchs \cite{FuchsR1905:quelques}.
At present, as with the Painleve equations, Garnier systems have been studied from various perspectives, including B\"acklund transformations \cite{KimuraH1990:MR1078402,KO1984:MR776915,SuzukiT2005:MR2177118,TsudaT2003:MR1998470,TsudaT2003:MR1987136}, special solutions \cite{OK1986:zbMATH03962350,TsudaT2003:MR1998470,TsudaT2006:MR2263717,YamadaY2009:zbMATH05588071}, and applications to discrete holomorphic functions \cite{JKMNS2017:MR3741826}.
See, {\it e.g.}, \cite{TakenawaT2024:RIMS} for details on Garnier systems.

In this paper, we focus in particular on P$_{\rm III}$, P$_{\rm VI}$, and the Garnier system in two variables.
\subsection{Notation and Terminology}
In this subsection, we summarize the notation and terminology used in this paper.
\begin{itemize}
\item
We denote the $J$-th Painlev\'e equation by {\rm P}$_{\rm J}$, and the Garnier system in $2$-variables by {\rm Garnier}$_{\rm 2}$.
\item
For a given $A \in \mathbb{C}$, when the notation $A^{M/N}$ with $M, N \in \mathbb{Z}$ is used, we choose one branch of $A^{1/|N|}$ such that its $|N|$-th power equals $A$.  
If both notations $A^{M_1/N_1}$ with $M_1, N_1 \in \mathbb{Z}$ and $A^{M_2/N_2}$ with $M_2, N_2 \in \mathbb{Z}$ appear simultaneously, we choose one branch of $A^{1/N_3}$ such that its $N_3$-th power equals $A$, where $N_3$ is the least common multiple of $|N_1|$ and $|N_2|$.
\item
For matrices $A$ and $B$, the symbol $AB$ means the matrix product $A.B$.
\item 
For transformations $s$ and $r$, the symbol $sr$ means the composite transformation $s\circ r$.
\item 
In the context of transformations, the ``$1$" signifies the identity transformation.
\item
For transformation $s$, the relation $s^\infty=1$ implies that there is no positive integer $k$ such as
$s^k=1$.
\item 
If an equation number is accompanied by a subscript such as ``$l\to l+1$", it means that the variable $l$ in the corresponding equation should be replaced by $l+1$.
For example, for the equation
\begin{equation}\label{eqn:example1}
 a_l+2a_{l+1}+3a_{l+2}=0,
\end{equation}
the notation \eqref{eqn:example1}$_{l\to l+1}$ represents
\begin{equation}
 a_{l+1}+2a_{l+2}+3a_{l+3}=0.
\end{equation}
The same convention applies when subscripts such as ``$l=0$" or other symbols are used.
For instance, the notation \eqref{eqn:example1}$_{l=0}$ is equivalent to
\begin{equation}
 a_0+2a_1+3a_2=0.
\end{equation}
\end{itemize}

\subsection{Outline of the paper}
This paper is organized as follows.
In \S \ref{section:main}, we present the main results showing that Equations \eqref{eqn:auto_dkdv}--\eqref{eqn:auto_dV} admit special solutions via {\rm P}$_{\rm III}$, {\rm P}$_{\rm VI}$, or {\rm Garnier}$_{\rm 2}$, depending on the equation.
In \S \ref{section:Backlund_P3P6Garnier}, we explain how the transformations $T_i$ used in \S \ref{section:main} are derived from the symmetry groups formed by B\"acklund transformations of {\rm P}$_{\rm III}$, {\rm P}$_{\rm VI}$, and {\rm Garnier}$_{\rm 2}$.
Concluding remarks are given in \S \ref{ConcludingRemarks}.
Appendix \ref{appendix:additive_dKdV} shows that a non-autonomous dKdV equation admits a special solution via {\rm P}$_{\rm III}$. 
Appendix \ref{appendix:CAC} shows that the HV equation \eqref{eqn:auto_HV} has the CAC property, and Appendix \ref{appendix:CABC} shows that the dVolterra equation \eqref{eqn:auto_dV} has the CABC property.

\section{Main results}\label{section:main}
In this section, we show that Equations \eqref{eqn:auto_dkdv}--\eqref{eqn:auto_dV} admit special solutions via {\rm P}$_{\rm III}$, {\rm P}$_{\rm VI}$, or {\rm Garnier}$_{\rm 2}$, depending on the equation.
The correspondence between the special solutions presented in this paper is summarized in Table \ref{table:solutions}.
How the transformations $T_i$ used in this section are derived from the symmetry groups (that is, the groups of B\"acklund transfomrations) of {\rm P}$_{\rm III}$, {\rm P}$_{\rm VI}$, and {\rm Garnier}$_{\rm 2}$ is explained in \S \ref{section:Backlund_P3P6Garnier}.

\begin{table}[htp]
\begin{center}
\begin{tabular}{|l||c|c|c|}
\hline
&{\rm P}$_{\rm III}$&{\rm P}$_{\rm VI}$&{\rm Garnier}$_{\rm 2}$\\[-0.5em]
&(See \S \ref{subsection:main_P3})&(See \S \ref{subsection:main_P6})&(See \S \ref{subsection:main_Garnier})\\
\hhline{|=||=|=|=|}
The dKdV equation \eqref{eqn:auto_dkdv}&$\bigcirc$&$\times$&$\times$\\
\hline
The $Q1_{\de=1}$ equation \eqref{eqn:auto_Q1}&$\times$&$\bigcirc$&$\times$\\
\hline
The HV equation \eqref{eqn:auto_HV}&$\bigcirc$&$\bigcirc$&$\times$\\
\hline
The lsG equation \eqref{eqn:auto_lsG}&$\times$&$\bigcirc$&$\bigcirc$\\
\hline
The dVolterra equation \eqref{eqn:auto_dV}&$\bigcirc$&$\bigcirc$&$\bigcirc$\\
\hline
\end{tabular}
\end{center}
\caption{Correspondence between the special solutions presented in this paper. Note that the symbol ``$\times$" in the table does not indicate nonexistence, but rather that such solutions have not been found in the present study.}
\label{table:solutions}
\end{table}

\subsection{The third Painlev\'e equation}\label{subsection:main_P3}
Let $t\in\bbC$ be the independent variable, $p=p(t)\in\bbC$ and $q=q(t)\in\bbC$ dependent variables, and $a_1,a_2\in\bbC$ parameters.
{\rm P}$_{\rm III}$ is equivalent to the Hamiltonian system:
\begin{equation}\label{eqn:Hamiltonian_system_P3}
 \dfrac{d q}{d t}=\dfrac{\partial H_{\rm III}}{\partial p},\quad
 \dfrac{d p}{d t}=-\dfrac{\partial H_{\rm III}}{\partial q},
\end{equation}
with the Hamiltonian \cite{KNY2017:MR3609039,OkamotoK1987:MR927186}:
\begin{equation}
 H_{\rm III}=H(p,q,t\,;a_1,a_2)=\dfrac{1}{t}\left(p(p-1)q^2+(a_1+a_2)pq+tp-a_2q\right).
\end{equation}
We define the birational transformations $T_1$ and $T_2$ by their actions on $\{a_1,a_2,t,p,q\}$ as follows:
\begin{subequations}\label{eqns:action_T1T2_P3}
\begin{align}
 &T_1:(a_1,a_2,t)\mapsto (a_1+1,a_2,t),\\
 &T_2:(a_1,a_2,t)\mapsto (a_1,a_2+1,t),\\
 &T_1(p)=-\dfrac{tp}{T_1(q)^2}-\dfrac{a_2}{T_1(q)},\quad
 T_1(q)=\dfrac{t(p-1)}{a_1+(p-1)q},\\
 &T_2(p)=\dfrac{t(p-1)}{T_2(q)^2}-\dfrac{a_1}{T_2(q)}+1,\quad
 T_2(q)=-\dfrac{tp}{pq+a_2}.
\end{align}
\end{subequations}
By direct computation, these transformations satisfy the following relations as actions on $\{a_1,a_2,t,p,q\}$:
\begin{equation}
 T_1\,\dfrac{d}{d t}=\dfrac{d}{d t}\,T_1,\quad
 T_2\,\dfrac{d}{d t}=\dfrac{d}{d t}\,T_2,\quad
 T_1T_2=T_2T_1.
\end{equation}
Hence, the transformations $T_1$ and $T_2$ are mutually commuting B\"acklund transformations of {\rm P}$_{\rm III}$ \eqref{eqn:Hamiltonian_system_P3}.
Accordingly, we define $a_1^{(k_1)}$ and $a_2^{(k_2)}$ by
\begin{equation}
 a_1^{(k_1)}={T_1}^{k_1}(a_1),\quad
 a_2^{(k_2)}={T_2}^{k_2}(a_2),
\end{equation}
and define $p^{(k_1,k_2)}$, $q^{(k_1,k_2)}$, and $H_{\rm III}^{(k_1,k_2)}$ by
\begin{subequations}
\begin{align}
 &p^{(k_1,k_2)}={T_1}^{k_1}{T_2}^{k_2}(p),\quad
 q^{(k_1,k_2)}={T_1}^{k_1}{T_2}^{k_2}(q),\\
 &H_{\rm III}^{(k_1,k_2)}=H\left(p^{(k_1,k_2)},q^{(k_1,k_2)},t\,;a_1^{(k_1)},a_2^{(k_2)}\right),
\end{align}
\end{subequations}
so that
\begin{equation}
 \dfrac{d q^{(k_1,k_2)}}{d t}=\dfrac{\partial H_{\rm III}^{(k_1,k_2)}}{\partial p^{(k_1,k_2)}},\quad
 \dfrac{d p^{(k_1,k_2)}}{d t}=-\dfrac{\partial H_{\rm III}^{(k_1,k_2)}}{\partial q^{(k_1,k_2)}},
\end{equation}
is again equivalent to {\rm P}$_{\rm III}$.
Note that, in the present setting, {\rm P}$_{\rm III}$ is regarded as an equation with independent variable $t$, dependent variables $p^{(k_1,k_2)}$ and $q^{(k_1,k_2)}$, and parameters $a_1^{(k_1)}$ and $a_2^{(k_2)}$.
Moreover, from the actions of $T_1$ and $T_2$ in \eqref{eqns:action_T1T2_P3}, we obtain the relations among the parameters:
\begin{equation}
 a_1^{(k_1+1)}=a_1^{(k_1)}+1,\quad
 a_2^{(k_2+1)}=a_2^{(k_2)}+1,
\end{equation}
the {\ODE} in the $k_1$-direction:
\begin{subequations}
\begin{align}
 &p^{(k_1+1,k_2)}
 =-\dfrac{tp^{(k_1,k_2)}}{{q^{(k_1+1,k_2)}}^2}-\dfrac{a_2^{(k_2)}}{q^{(k_1+1,k_2)}},\label{eqn:ODE_T1T2_P3_1}\\
 &q^{(k_1+1,k_2)}
 =\dfrac{t(p^{(k_1,k_2)}-1)}{a_1^{(k_1)}+q^{(k_1,k_2)}(p^{(k_1,k_2)}-1)},\label{eqn:ODE_T1T2_P3_2}
\end{align}
\end{subequations}
and the {\ODE} in the $k_2$-direction:
\begin{subequations}
\begin{align} 
 &p^{(k_1,k_2+1)}
 =\dfrac{t(p^{(k_1,k_2)}-1)}{{q^{(k_1,k_2+1)}}^2}-\dfrac{a_1^{(k_1)}}{q^{(k_1,k_2+1)}}+1,\label{eqn:ODE_T1T2_P3_3}\\
 &q^{(k_1,k_2+1)}
 =-\dfrac{tp^{(k_1,k_2)}}{p^{(k_1,k_2)}q^{(k_1,k_2)}+a_2^{(k_2)}}.\label{eqn:ODE_T1T2_P3_4}
\end{align}
\end{subequations}
We obtain the following theorem.
\begin{theorem}\label{theorem:P3solutions}
The following hold.
\begin{itemize}
\item[\rm (i)]
The dKdV equation \eqref{eqn:auto_dkdv} admits the following special solution:
\begin{equation}\label{eqn:dkdv_P3sol}
 u_{l,m}=\dfrac{q^{(l,m)}}{t^{1/2}}.
\end{equation}
\item[\rm (ii)]
The HV equation \eqref{eqn:auto_HV} admits the following special solutions:
\begin{subequations}\label{eqns:HV_P3sol}
\begin{align}
 &u_{l,m}=-\dfrac{t\,\Big(H_{\rm III}^{(l,m)}-1\Big)}{a_1^{(l)}},\\
 &u_{l,m}=-\dfrac{t\,H_{\rm III}^{(m,l)}}{a_2^{(l)}}.
\end{align}
\end{subequations}
\item[\rm (iii)]
The dVolterra equation \eqref{eqn:auto_dV} admits the following special solutions:
\begin{subequations}\label{eqns:LV_P3sol}
\begin{align}
 &u_{l,m}=-\dfrac{\Big(p^{(l,m)}-1\Big)q^{(l,m)}}{a_1^{(l)}}-1,\\
 &u_{l,m}=-\dfrac{p^{(m,l)}q^{(m,l)}}{a_2^{(l)}}-1.
\end{align}
\end{subequations} 
\end{itemize}
\end{theorem}
\begin{proof}
We prove (i).
Substituting \eqref{eqn:dkdv_P3sol} into Equation \eqref{eqn:auto_dkdv}, we obtain
\begin{equation}\label{eqn:proof_P3_dKdV_1}
 q^{(l+1,m+1)}-q^{(l,m)}=\dfrac{~t~}{q^{(l,m+1)}}-\dfrac{~t~}{q^{(l+1,m)}}.
\end{equation}
From Equation \eqref{eqn:proof_P3_dKdV_1}, eliminating $q^{(l+1,m+1)}$ and $q^{(l+1,m)}$ by using Equations \eqref{eqn:ODE_T1T2_P3_2}$_{(k_1,k_2)=(l,m+1)}$ and \eqref{eqn:ODE_T1T2_P3_2}$_{(k_1,k_2)=(l,m)}$, we obtain Equation \eqref{eqn:ODE_T1T2_P3_3}$_{(k_1,k_2)=(l,m)}$.
Therefore, \eqref{eqn:dkdv_P3sol} is a special solution of Equation \eqref{eqn:auto_dkdv}.

Statements (ii) and (iii) can be verified by direct computation in the same manner.
This completes the proof.
\end{proof}

\subsection{The sixth Painlev\'e equation}\label{subsection:main_P6}
Let $t\in\bbC$ be the independent variable, $p=p(t)\in\bbC$ and $q=q(t)\in\bbC$ dependent variables, and $a_0,\dots,a_4\in\bbC$ parameters satisfying
\begin{equation}\label{eqn:condition_parameters_P6}
 a_0+a_1+2a_2+a_3+a_4=1.
\end{equation}
{\rm P}$_{\rm VI}$ is equivalent to the Hamiltonian system:
\begin{equation}\label{eqn:Hamiltonian_system_P6}
 \dfrac{d q}{d t}=\dfrac{\partial H_{\rm VI}}{\partial p},\quad
 \dfrac{d p}{d t}=-\dfrac{\partial H_{\rm VI}}{\partial q},
\end{equation}
with the Hamiltonian \cite{OkamotoK1987:MR916698}:
\begin{align}
 H_{\rm VI}
 =&H(p,q,t\,;a_0,a_1,a_2,a_3,a_4)\notag\\
 =&\dfrac{1}{t(t-1)}\Bigg(q(q-1)(q-t)p^2+a_2(a_1+a_2)(q-t)\notag\\
 &\hspace{3.5em}-\Big(a_4(q-1)(q-t)+a_3q(q-t)+(a_0-1)q(q-1)\Big)p\Bigg).
\end{align}
For simplicity, we define the set ${\mathcal S}_{\rm VI}$ by
\begin{equation}
 {\mathcal S}_{\rm VI}=\{a_0,a_1,a_2,a_3,a_4,t,p,q\}.
\end{equation}
We define the birational transformations $T_i$, $i=1,2,3$, by their actions on ${\mathcal S}_{\rm VI}$ as follows:
\begin{subequations}\label{eqns:action_T1T2T3_P6}
\begin{align}
 &T_1:(a_0,a_1,a_2,a_3,a_4,t)\notag\\
 &\hspace{2em}\mapsto(1-a_1-a_2,1-a_0-a_2,a_2-1,1-a_2-a_4,1-a_2-a_3,t),\\
 &T_2:(a_0,a_1,a_2,a_3,a_4,t)\notag\\
 &\hspace{2em}\mapsto(1-a_1-a_2,-a_0-a_2,a_2,1-a_2-a_4,-a_2-a_3,t),\\
 &T_3:(a_0,a_1,a_2,a_3,a_4,t)\notag\\
 &\hspace{2em}\mapsto(1-a_1-a_2,-a_0-a_2,a_2,-a_2-a_4,1-a_2-a_3,t),\\
 &T_1(p)=-\dfrac{(t-q)(pq-tp+a_1+a_2)}{(1-t)t}-\dfrac{a_4(t-q)}{(1-t)q}-\dfrac{a_3(t-q)}{t(1-q)},\\
 &T_1(q)=\dfrac{t\,\Big(T_1(p)-a_2+1\Big)+\Big(a_2-1-t\,T_1(p)\Big)q}{(t-q)T_1(p)},\\
 &T_2(p)=\dfrac{(t-q)(pq-p+a_2)}{(1-t)T_2(q)},\quad
 T_2(q)=\cfrac{(t-q)p-a_2-\dfrac{a_3(1-t)}{1-q}}{(1-q)p-a_2+\dfrac{a_0(1-t)}{t-q}},\\
 &T_3(p)=-\cfrac{q(pq+a_2)\left(pq+a_2+\dfrac{a_0t}{t-q}\right)}{t\left(pq+a_2-\dfrac{(a_2+a_4)t}{t-q}\right)},\quad
 T_3(q)=\cfrac{t\left(pq-a_0-a_4+\dfrac{a_0t}{t-q}\right)}{q\left(pq+a_2+\dfrac{a_0t}{t-q}\right)}.
\end{align}
\end{subequations}
By direct computation, these transformations satisfy the following relations as actions on the set ${\mathcal S}_{\rm VI}$:
\begin{equation}
 T_i\,\dfrac{d}{d t}=\dfrac{d}{d t}\,T_i,\quad
 T_iT_j=T_jT_i,
\end{equation}
where $i,j=1,2,3$.
Hence, the transformations $T_i$, $i=1,2,3$, are mutually commuting B\"acklund transformations of {\rm P}$_{\rm VI}$ \eqref{eqn:Hamiltonian_system_P6}.
Accordingly, we define $a_i^{(k_1,k_2,k_3)}$ for $i=0,1,3,4$, and $a_2^{(k_1)}$ by
\begin{equation}
 a_i^{(k_1,k_2,k_3)}={T_1}^{k_1}{T_2}^{k_2}{T_3}^{k_3}(a_i),\quad i=0,1,3,4,\qquad
 a_2^{(k_1)}={T_1}^{k_1}(a_2),
\end{equation}
and define $p^{(k_1,k_2,k_3)}$, $q^{(k_1,k_2,k_3)}$, and $H_{\rm VI}^{(k_1,k_2,k_3)}$ by
\begin{subequations}
\begin{align}
 &p^{(k_1,k_2,k_3)}={T_1}^{k_1}{T_2}^{k_2}{T_3}^{k_3}(p),\quad
 q^{(k_1,k_2,k_3)}={T_1}^{k_1}{T_2}^{k_2}{T_3}^{k_3}(q),\\
 &H_{\rm VI}^{(k_1,k_2,k_3)}=H\left(p^{(k_1,k_2,k_3)},q^{(k_1,k_2,k_3)},t\,;a_0^{(k_1,k_2,k_3)},a_1^{(k_1,k_2,k_3)},a_2^{(k_1)},a_3^{(k_1,k_2,k_3)},a_4^{(k_1,k_2,k_3)}\right),
\end{align}
\end{subequations}
so that
\begin{equation}
 \dfrac{d q^{(k_1,k_2,k_3)}}{d t}=\dfrac{\partial H_{\rm VI}^{(k_1,k_2,k_3)}}{\partial p^{(k_1,k_2,k_3)}},\quad
 \dfrac{d p^{(k_1,k_2,k_3)}}{d t}=-\dfrac{\partial H_{\rm VI}^{(k_1,k_2,k_3)}}{\partial q^{(k_1,k_2,k_3)}},
\end{equation}
is again equivalent to {\rm P}$_{\rm VI}$.
Note that, from the condition \eqref{eqn:condition_parameters_P6}, we obtain 
\begin{equation}
 a_0^{(k_1,k_2,k_3)}+a_1^{(k_1,k_2,k_3)}+2a_2^{(k_1)}+a_3^{(k_1,k_2,k_3)}+a_4^{(k_1,k_2,k_3)}=1,
\end{equation}
and from the actions of $T_i$, $i=1,2,3$, in \eqref{eqns:action_T1T2T3_P6}, we obtain 
\begin{subequations}
\begin{align}
 &a_2^{(k_1+1)}=a_2^{(k_1)}-1,&&\\
 &a_0^{(k_1+1,k_2,k_3)}=1-a_1^{(k_1,k_2,k_3)}-a_2^{(k_1)},
 &&a_1^{(k_1+1,k_2,k_3)}=1-a_0^{(k_1,k_2,k_3)}-a_2^{(k_1)},\\
 &a_3^{(k_1+1,k_2,k_3)}=1-a_2^{(k_1)}-a_4^{(k_1,k_2,k_3)},
 &&a_4^{(k_1+1,k_2,k_3)}=1-a_2^{(k_1)}-a_3^{(k_1,k_2,k_3)},\\[0.5em]
 &a_0^{(k_1,k_2+1,k_3)}=1-a_1^{(k_1,k_2,k_3)}-a_2^{(k_1)},
 &&a_1^{(k_1,k_2+1,k_3)}=-a_0^{(k_1,k_2,k_3)}-a_2^{(k_1)},\\
 &a_3^{(k_1,k_2+1,k_3)}=1-a_2^{(k_1)}-a_4^{(k_1,k_2,k_3)},
 &&a_4^{(k_1,k_2+1,k_3)}=-a_2^{(k_1)}-a_3^{(k_1,k_2,k_3)},\\[0.5em]
 &a_0^{(k_1,k_2,k_3+1)}=1-a_1^{(k_1,k_2,k_3)}-a_2^{(k_1)},
 &&a_1^{(k_1,k_2,k_3+1)}=-a_0^{(k_1,k_2,k_3)}-a_2^{(k_1)},\\
 &a_3^{(k_1,k_2,k_3+1)}=-a_2^{(k_1)}-a_4^{(k_1,k_2,k_3)},
 &&a_4^{(k_1,k_2,k_3+1)}=1-a_2^{(k_1)}-a_3^{(k_1,k_2,k_3)}.
\end{align}
\end{subequations}
We obtain the following theorem.
\begin{theorem}\label{theorem:P6solutions}
The following hold.
\begin{itemize}
\item[\rm (i)]
The $Q1_{\de=1}$ equation \eqref{eqn:auto_Q1} admits the following special solution:
\begin{equation}\label{eqn:Q1_P6sol}
 u_{l,m}=\dfrac{2t(t-1)^2}{a_2}\left(H_{\rm VI}^{(0,l,m)}+\dfrac{a_2\left(a_1^{(0,l,m)}+a_3^{(0,l,m)}\right)}{2t}+\dfrac{a_2\left(a_1^{(0,l,m)}+a_4^{(0,l,m)}\right)}{2(t-1)}\right),
\end{equation}
with
\begin{equation}
 \al=(t-1)t,\quad
 \be=(t-1)^2.
\end{equation}
\item[\rm (ii)]
The HV equation \eqref{eqn:auto_HV} admits the following special solutions:
\begin{subequations}\label{eqns:HV_P6sol}
\begin{align}
 &u_{l,m}=\dfrac{(t-1)H_{\rm VI}^{(l,m,0)}}{a_2^{(l)}},\\
 &u_{l,m}=\dfrac{t\,H_{\rm VI}^{(l,0,m)}}{a_2^{(l)}}.
\end{align}
\end{subequations}
\item[\rm (iii)]
The lsG equation \eqref{eqn:auto_lsG} admits the following special solutions:
\begin{subequations}\label{eqns:lsG_P6sol}
\begin{align}
 &u_{l,m}=\dfrac{q^{(l,m,0)}}{t^{1/2}}\quad \text{with}\quad \al=t^{1/2}\be,\\
 &u_{l,m}=\dfrac{1-q^{(l,0,m)}}{(1-t)^{1/2}}\quad \text{with}\quad \al=(1-t)^{1/2}\be.
\end{align}
\end{subequations}
\item[\rm (iv)]
The dVolterra equation \eqref{eqn:auto_dV} admits the following special solution:
\begin{subequations}\label{eqns:LV_P6sol}
\begin{align}
 &u_{l,m}=-\dfrac{q^{(l,m,0)}\Big(p^{(l,m,0)}q^{(l,m,0)}-tp^{(l,m,0)}+a_2^{(l)}\Big)}{a_2^{(l)} t},\\
 &u_{l,m}=-\dfrac{p^{(l,0,m)}q^{(l,0,m)}}{a_2^{(l)}}-1,\\
 &u_{l,m}=\dfrac{p^{(0,l,m)}q^{(0,l,m)}-a_4^{(0,l,m)}}{a_1^{(0,l,m)}+a_2+a_4^{(0,l,m)}}.
\end{align}
\end{subequations} 
\end{itemize}
\end{theorem}
\begin{proof}
We prove (i).
From the actions of the transformations $T_2$ and $T_3$ in \eqref{eqns:action_T1T2T3_P6}, we obtain the following two {\ODE}s:
\begin{align}
 &\begin{cases}\label{eqn:ODE_T2T3_P6_1}
 ~p^{(0,l+1,m)}=\dfrac{\Big(t-q^{(0,l,m)}\Big)\Big(p^{(0,l,m)}q^{(0,l,m)}-p^{(0,l,m)}+a_2\Big)}{(1-t)q^{(0,l+1,m)}},\\[1em]
 ~q^{(0,l+1,m)}=\cfrac{\Big(t-q^{(0,l,m)}\Big)p^{(0,l,m)}-a_2-\dfrac{a_3^{(0,l,m)}(1-t)}{1-q^{(0,l,m)}}}{\Big(1-q^{(0,l,m)}\Big)p^{(0,l,m)}-a_2+\dfrac{a_0^{(0,l,m)}(1-t)}{t-q^{(0,l,m)}}},
 \end{cases}\\
 &\begin{cases}\label{eqn:ODE_T2T3_P6_2}
 ~p^{(0,l,m+1)}=-\cfrac{q^{(0,l,m)}(p^{(0,l,m)}q^{(0,l,m)}+a_2)\left(p^{(0,l,m)}q^{(0,l,m)}+a_2+\dfrac{a_0^{(0,l,m)}t}{t-q^{(0,l,m)}}\right)}{t\left(p^{(0,l,m)}q^{(0,l,m)}+a_2-\dfrac{(a_2+a_4^{(0,l,m)})t}{t-q^{(0,l,m)}}\right)},\\[3em]
 ~q^{(0,l,m+1)}=\cfrac{t\left(p^{(0,l,m)}q^{(0,l,m)}-a_0^{(0,l,m)}-a_4^{(0,l,m)}+\dfrac{a_0^{(0,l,m)}t}{t-q^{(0,l,m)}}\right)}{q^{(0,l,m)}\left(p^{(0,l,m)}q^{(0,l,m)}+a_2+\dfrac{a_0^{(0,l,m)}t}{t-q^{(0,l,m)}}\right)}.
 \end{cases}
\end{align}
Substituting \eqref{eqn:Q1_P6sol} into Equation \eqref{eqn:auto_Q1}, we obtain a relation among
\begin{equation}
 \Big\{p^{(0,l,m)},p^{(0,l+1,m)},p^{(0,l,m+1)},p^{(0,l+1,m+1)},q^{(0,l,m)},q^{(0,l+1,m)},q^{(0,l,m+1)},q^{(0,l+1,m+1)}\Big\}.
\end{equation}
From the resulting relation, eliminating $p^{(0,l+1,m+1)}$ and $q^{(0,l+1,m+1)}$ by using Equation \eqref{eqn:ODE_T2T3_P6_1}$_{m\to m+1}$, we obtain a relation among
\begin{equation}
 \Big\{p^{(l,m)},p^{(l+1,m)},p^{(l,m+1)},q^{(l,m)},q^{(l+1,m)},q^{(l,m+1)}\Big\}.
\end{equation}
Furthermore, eliminating $p^{(0,l+1,m)}$ and $q^{(0,l+1,m)}$ from the resulting relation by using Equation \eqref{eqn:ODE_T2T3_P6_1}, we obtain a relation among
\begin{equation}
 \Big\{p^{(l,m)},p^{(l,m+1)},q^{(l,m)},q^{(l,m+1)}\Big\}.
\end{equation}
By Equation \eqref{eqn:ODE_T2T3_P6_2}, we can finally verify that the obtained relation is satisfied.
Therefore, \eqref{eqn:Q1_P6sol} is a special solution of Equation \eqref{eqn:auto_Q1}.

Statements (ii), (iii), and (iv) can be verified by direct computation in the same manner.
This completes the proof.
\end{proof}

\subsection{The Garnier system in two variables}\label{subsection:main_Garnier}
Let $t_1,t_2\in\bbC$ be independent variables, $p_i=p_i(t_1,t_2)\in\bbC$ and $q_j=q_j(t_1,t_2)\in\bbC$ $(i,j=1,2)$ dependent variables, and $\theta_1,\theta_2,\ka_0,\ka_1,\ka_\infty\in\bbC$ parameters.
{\rm Garnier}$_{\rm 2}$ is equivalent to the Hamiltonian system:
\begin{equation}\label{eqn:Hamiltonian_system_Garnier}
 \dfrac{\partial q_j}{\partial t_i}=\dfrac{\partial H_i}{\partial p_j},\quad
 \dfrac{\partial p_j}{\partial t_i}=-\dfrac{\partial H_i}{\partial q_j},
\end{equation}
where $i,j=1,2$, with the Hamiltonians \cite{KO1984:MR776915}:
\begin{align}
H_i
=&H_i(p_1,p_2,q_1,q_2,t_1,t_2\,;\theta_1,\theta_2,\ka_0,\ka_1,\ka_\infty)\notag\\
=&\dfrac{1}{t_i(t_i-1)}\Big(q_i(p_1q_1+p_2q_2+\ga)(p_1q_1+p_2q_2+\ga+\ka_{\infty})+t_ip_i(p_iq_i-\theta_i)\notag\\
&\hspace{4em}-\dfrac{t_j(t_i-1)}{t_i-t_j}(p_jq_j-\theta_j)p_jq_i
-\dfrac{t_i(t_i-1)}{t_i-t_j}(p_iq_i-\theta_i)p_iq_j\notag\\
&\hspace{4em}-\dfrac{t_i(t_j-1)}{t_j-t_i}(p_iq_i-\theta_i)p_jq_j
-\dfrac{t_i(t_j-1)}{t_j-t_i}(p_jq_j-\theta_j)p_iq_i\notag\\
&\hspace{4em}-(t_i+1)(p_iq_i-\theta_i)p_iq_i+(\ka_1t_i+\ka_0-1)p_iq_i\Big),
\end{align}
where $(i,j)=(1,2),(2,1)$ and 
\begin{equation}\label{eqn:def_gamma}
 \ga=-\dfrac{\theta_1+\theta_2+\ka_0+\ka_1+\ka_\infty-1}{2}.
\end{equation} 
For simplicity, we define the set ${\mathcal S}_{\rm Garnier}$ by
\begin{equation}
 {\mathcal S}_{\rm Garnier}=\{\theta_1,\theta_2,\ka_0,\ka_1,\ka_\infty,t_1,t_2,p_1,p_2,q_1,q_2\}.
\end{equation}
We define the birational transformations $T_i$, $i=1,2,3$, by their actions on ${\mathcal S}_{\rm Garnier}$ as follows:
\begin{subequations}\label{eqns:action_T1T2T3_Garnier}
\begin{align}
 &T_1:(\theta_1,\theta_2,\ka_0,\ka_1,\ka_\infty,t_1,t_2)\mapsto(\theta_1+1,\theta_2,\ka_0,1-\ka_1,\ka_\infty,t_1,t_2),\\
 &T_2:(\theta_1,\theta_2,\ka_0,\ka_1,\ka_\infty,t_1,t_2)\mapsto(\theta_1,\theta_2+1,\ka_0,1-\ka_1,\ka_\infty,t_1,t_2),\\
 &T_3:(\theta_1,\theta_2,\ka_0,\ka_1,\ka_\infty,t_1,t_2)\mapsto(\theta_1,\theta_2,\ka_0+1,1-\ka_1,\ka_\infty,t_1,t_2),\\
 &T_1(p_1)=\dfrac{(q_1f_3-t_1f_1)(q_1f_3-t_1f_1+\ka_\infty q_1)}{(1-t_1)t_1q_1f_1},\\
 &T_1(p_2)=\dfrac{t_1(f_1-p_2q_1)T_1(p_1)}{t_1f_1-t_2p_2q_1},\\
 &T_1(q_1)=\dfrac{\ka_0t_2q_1-f_1f_4}{t_2T_1(p_1)q_1}-\dfrac{t_1T_1(q_2)}{t_2}+t_1,\\
 &T_1(q_2)=\dfrac{(t_2p_2q_1-t_1f_1)(t_1q_2f_1-t_2q_1f_2)}{(t_1-t_2)t_1T_1(p_1)q_1f_1},\\
 &T_2(p_1)=\dfrac{t_2(f_2-p_1q_2)T_2(p_2)}{t_2f_2-t_1p_1q_2},\\
 &T_2(p_2)=\dfrac{(q_2f_3-t_2f_2)(q_2f_3-t_2f_2+\ka_\infty q_2)}{(1-t_2)t_2q_2f_2},\\
 &T_2(q_1)=\dfrac{(t_1p_1q_2-t_2f_2)(t_1q_2f_1-t_2q_1f_2)}{(t_1-t_2)t_2T_2(p_2)q_2f_2},\\
 &T_2(q_2)=\dfrac{\ka_0 t_1 q_2-f_2f_4}{t_1T_2(p_2)q_2}-\dfrac{t_2T_2(q_1)}{t_1}+t_2,\\
 &T_3(p_1)=\dfrac{f_5}{t_1(\ka_0t_2-p_1f_4)f_4},\\
 &T_3(p_2)=\dfrac{f_5}{t_2(\ka_0t_1-p_2f_4)f_4},\\
 &T_3(q_1)=\dfrac{t_1(\ka_0t_2-p_1f_4)(\ka_0t_2q_1-f_1f_4)}{f_5},\\
 &T_3(q_2)=\dfrac{t_2(\ka_0t_1-p_2f_4)(\ka_0t_1q_2-f_2f_4)}{f_5},
\end{align}
\end{subequations}
where
\begin{subequations}
\begin{align}
 &f_1=p_1q_1-\theta_1,\quad
 f_2=p_2q_2-\theta_2,\quad
 f_3=p_1q_1+p_2q_2+\ga,\\
 &f_4=q_1q_2-(q_1-t_1)(q_2-t_2),\quad
 f_5=(f_3f_4-\ka_0 t_1t_2)(f_3f_4+\ka_\infty f_4-\ka_0 t_1t_2).
\end{align}
\end{subequations}
By direct computation, these transformations satisfy the following relations as actions on the set ${\mathcal S}_{\rm Garnier}$:
\begin{equation}
 T_i\,\partial_{t_k}=\partial_{t_k}T_i,\quad
 T_iT_j=T_jT_i,\quad
\end{equation}
where $i,j=1,2,3$, $k=1,2$, and $\partial_{t_k}=\partial/\partial t_k$.
Hence, the transformations $T_i$, $i=1,2,3$, are mutually commuting B\"acklund transformations of {\rm Garnier}$_{\rm 2}$ \eqref{eqn:Hamiltonian_system_Garnier}.
Accordingly, we define $\theta_1^{(k_1)}$, $\theta_2^{(k_2)}$, $\ka_0^{(k_3)}$, and $\ka_1^{(k_1,k_2,k_3)}$ by
\begin{subequations}
\begin{align}
 &\theta_1^{(k_1)}={T_1}^{k_1}(\theta_1),\quad
 \theta_2^{(k_2)}={T_2}^{k_2}(\theta_2),\quad
 \ka_0^{(k_3)}={T_3}^{k_3}(\ka_0),\\
 &\ka_1^{(k_1,k_2,k_3)}={T_1}^{k_1}{T_2}^{k_2}{T_3}^{k_3}(\ka_1),
\end{align}
\end{subequations}
and define $p_i^{(k_1,k_2,k_3)}$, $q_j^{(k_1,k_2,k_3)}$, $H_k^{(k_1,k_2,k_3)}$ $(i,j,k=1,2)$ by
\begin{subequations}
\begin{align}
 &p_i^{(k_1,k_2,k_3)}={T_1}^{k_1}{T_2}^{k_2}{T_3}^{k_3}(p_i),\quad
 q_j^{(k_1,k_2,k_3)}={T_1}^{k_1}{T_2}^{k_2}{T_3}^{k_3}(q_j),\\
 &H_k^{(k_1,k_2,k_3)}
 =H_k\Big(p_1^{(k_1,k_2,k_3)},p_2^{(k_1,k_2,k_3)},q_1^{(k_1,k_2,k_3)},q_2^{(k_1,k_2,k_3)},t_1,t_2\,;\notag\\
 &\hspace{13em}
 \theta_1^{(k_1)},\theta_2^{(k_2)},\ka_0^{(k_3)},\ka_1^{(k_1,k_2,k_3)},\ka_\infty\Big),
\end{align}
\end{subequations}
so that
\begin{equation}
 \dfrac{\partial q_j^{(k_1,k_2,k_3)}}{\partial t_i}=\dfrac{\partial H_i^{(k_1,k_2,k_3)}}{\partial p_j^{(k_1,k_2,k_3)}},\quad
 \dfrac{\partial p_j^{(k_1,k_2,k_3)}}{\partial t_i}=-\dfrac{\partial H_i^{(k_1,k_2,k_3)}}{\partial q_j^{(k_1,k_2,k_3)}},\quad
 (i,j=1,2),
\end{equation}
is again equivalent to {\rm Garnier}$_{\rm 2}$.
Note that, from the actions of $T_i$, $i=1,2,3$, in \eqref{eqns:action_T1T2T3_Garnier}, we obtain 
\begin{subequations}
\begin{align}
 &\theta_1^{(k_1+1)}=\theta_1^{(k_1)}+1,\quad
 \theta_2^{(k_2+1)}=\theta_2^{(k_2)}+1,\quad
 \ka_0^{(k_3+1)}=\ka_0^{(k_3)}+1,\\
 &\ka_1^{(k_1+1,k_2,k_3)}=1-\ka_1^{(k_1,k_2,k_3)},\quad
 \ka_1^{(k_1,k_2+1,k_3)}=1-\ka_1^{(k_1,k_2,k_3)},\\
 &\ka_1^{(k_1,k_2,k_3+1)}=1-\ka_1^{(k_1,k_2,k_3)}.
\end{align}
\end{subequations}
We obtain the following theorem.
\begin{theorem}\label{theorem:Garniersolutions}
The following hold.
\begin{itemize}
\item[\rm (i)]
The lsG equation \eqref{eqn:auto_lsG} admits the following special solutions:
\begin{subequations}\label{eqns:lsG_Garniersol}
\begin{align}
 &u_{l,m}=\dfrac{{t_1}^{1/2}p_1^{(l,m,0)}}{{t_2}^{1/2}p_2^{(l,m,0)}}\quad 
 \text{with}\quad
 \al=\dfrac{{t_1}^{1/2}}{{t_2}^{1/2}}\,\be,\label{eqn:lsG_Garniersol_1}\\
 &u_{l,m}=-\dfrac{1}{{t_1}^{1/2}(t_1-t_2)^{1/2}}\left(t_2\dfrac{p_2^{(0,l,m)}}{p_1^{(0,l,m)}}-t_1\right)\quad 
 \text{with}\quad
 \al=\dfrac{(t_1-t_2)^{1/2}}{{t_1}^{1/2}}\,\be,\\
 &u_{l,m}=-\dfrac{1}{{t_2}^{1/2}(t_2-t_1)^{1/2}}\left(t_1\dfrac{p_1^{(l,0,m)}}{p_2^{(l,0,m)}}-t_2\right)\quad 
 \text{with}\quad
 \al=\dfrac{(t_2-t_1)^{1/2}}{{t_2}^{1/2}}\,\be.
\end{align}
\end{subequations}
\item[\rm (ii)]
The dVolterra equation \eqref{eqn:auto_dV} admits the following special solution:
\begin{subequations}\label{eqns:LV_Garniersol}
\begin{align}
 &u_{l,m}=\dfrac{p_1^{(l,0,m)}q_1^{(l,0,m)}}{\theta_1^{(l)}}-1,\\
 &u_{l,m}=\dfrac{p_2^{(0,l,m)}q_2^{(0,l,m)}}{\theta_2^{(l)}}-1,\\
 &u_{l,m}=\dfrac{p_1^{(m,0,l)}}{\ka_0^{(l)}t_2}\left(t_1q_2^{(m,0,l)}+t_2q_1^{(m,0,l)}-t_1t_2\right)-1,\\
 &u_{l,m}=\dfrac{p_2^{(0,m,l)}}{\ka_0^{(l)}t_1}\left(t_1q_2^{(0,m,l)}+t_2q_1^{(0,m,l)}-t_1t_2\right)-1.
\end{align}
\end{subequations} 
\end{itemize}
\end{theorem}
\begin{proof}
We prove that the first solution in (i) is correct.
From the actions of the transformations $T_1$ and $T_2$ in \eqref{eqns:action_T1T2T3_Garnier}, we obtain the following two {\ODE}s:
\begin{align}
 &\begin{cases}\label{eqn:ODE_T1T2_Garnier_1}
 ~p_1^{(l+1,m,0)}=\dfrac{\left(q_1^{(l,m,0)}f_3^{(l,m,0)}-t_1f_1^{(l,m,0)}\right)\left(q_1^{(l,m,0)}f_3^{(l,m,0)}-t_1f_1^{(l,m,0)}+\ka_\infty q_1^{(l,m,0)}\right)}{(1-t_1)t_1q_1^{(l,m,0)}f_1^{(l,m,0)}},\\[1em]
 ~p_2^{(l+1,m,0)}=\dfrac{t_1\left(f_1^{(l,m,0)}-p_2^{(l,m,0)}q_1^{(l,m,0)}\right)p_1^{(l+1,m,0)}}{t_1f_1^{(l,m,0)}-t_2p_2^{(l,m,0)}q_1^{(l,m,0)}},\\[1em]
 ~q_1^{(l+1,m,0)}=\dfrac{\ka_0t_2q_1^{(l,m,0)}-f_1^{(l,m,0)}f_4^{(l,m,0)}}{t_2p_1^{(l+1,m,0)}q_1^{(l,m,0)}}-\dfrac{t_1q_2^{(l+1,m,0)}}{t_2}+t_1,\\[1em]
 ~q_2^{(l+1,m,0)}=\dfrac{\left(t_2p_2^{(l,m,0)}q_1^{(l,m,0)}-t_1f_1^{(l,m,0)}\right)\left(t_1q_2^{(l,m,0)}f_1^{(l,m,0)}-t_2q_1^{(l,m,0)}f_2^{(l,m,0)}\right)}{(t_1-t_2)t_1p_1^{(l+1,m,0)}q_1^{(l,m,0)}f_1^{(l,m,0)}},
 \end{cases}\\
 &\begin{cases}\label{eqn:ODE_T1T2_Garnier_2}
 ~p_1^{(l,m+1,0)}=\dfrac{t_2\left(f_2^{(l,m,0)}-p_1^{(l,m,0)}q_2^{(l,m,0)}\right)p_2^{(l,m+1,0)}}{t_2f_2^{(l,m,0)}-t_1p_1^{(l,m,0)}q_2^{(l,m,0)}},\\[1em]
 ~p_2^{(l,m+1,0)}=\dfrac{\left(q_2^{(l,m,0)}f_3^{(l,m,0)}-t_2f_2^{(l,m,0)}\right)\left(q_2^{(l,m,0)}f_3^{(l,m,0)}-t_2f_2^{(l,m,0)}+\ka_\infty q_2^{(l,m,0)}\right)}{(1-t_2)t_2q_2^{(l,m,0)}f_2^{(l,m,0)}},\\[1em]
 ~q_1^{(l,m+1,0)}=\dfrac{\left(t_1p_1^{(l,m,0)}q_2^{(l,m,0)}-t_2f_2^{(l,m,0)}\right)\left(t_1q_2^{(l,m,0)}f_1^{(l,m,0)}-t_2q_1^{(l,m,0)}f_2^{(l,m,0)}\right)}{(t_1-t_2)t_2p_2^{(l,m+1,0)}q_2^{(l,m,0)}f_2^{(l,m,0)}},\\[1em]
 ~q_2^{(l,m+1,0)}=\dfrac{\ka_0 t_1 q_2^{(l,m,0)}-f_2^{(l,m,0)}f_4^{(l,m,0)}}{t_1p_2^{(l,m+1,0)}q_2^{(l,m,0)}}-\dfrac{t_2q_1^{(l,m+1,0)}}{t_1}+t_2,
 \end{cases}
\end{align}
where
\begin{subequations}
\begin{align}
 &f_1^{(l,m,0)}=p_1^{(l,m,0)}q_1^{(l,m,0)}-\theta_1^{(l)},\quad
 f_2^{(l,m,0)}=p_2^{(l,m,0)}q_2^{(l,m,0)}-\theta_2^{(m)},\\
 &f_3^{(l,m,0)}=p_1^{(l,m,0)}q_1^{(l,m,0)}+p_2^{(l,m,0)}q_2^{(l,m,0)}
 -\dfrac{\theta_1^{(l)}+\theta_2^{(m)}+\ka_0+\ka_1^{(l,m,0)}+\ka_\infty-1}{2},\\
 &f_4^{(l,m,0)}=q_1^{(l,m,0)}q_2^{(l,m,0)}-\left(q_1^{(l,m,0)}-t_1\right)\left(q_2^{(l,m,0)}-t_2\right).
\end{align}
\end{subequations}
Substituting \eqref{eqn:lsG_Garniersol_1} into Equation \eqref{eqn:auto_lsG}, we obtain a relation among
\begin{equation}
 \left\{
 p_1^{(l,m,0)},p_1^{(l+1,m,0)},p_1^{(l,m+1,0)},p_1^{(l+1,m+1,0)},
 p_2^{(l,m,0)},p_2^{(l+1,m,0)},p_2^{(l,m+1,0)},p_2^{(l+1,m+1,0)}
 \right\}.
\end{equation}
From the obtained relation, eliminating $p_1^{(l+1,m+1,0)}$ and $p_2^{(l+1,m+1,0)}$ by using Equation \eqref{eqn:ODE_T1T2_Garnier_1}$_{m\to m+1}$, we obtain a relation among
\begin{equation}
 \left\{
 p_1^{(l,m,0)},p_1^{(l+1,m,0)},p_1^{(l,m+1,0)},
 p_2^{(l,m,0)},p_2^{(l+1,m,0)},p_2^{(l,m+1,0)},
 q_1^{(l,m+1,0)},
 q_2^{(l,m+1,0)}
 \right\}.
\end{equation}
Furthermore, eliminating $p_1^{(l+1,m,0)}$ and $p_2^{(l+1,m,0)}$ from the obtained relation by using Equation \eqref{eqn:ODE_T1T2_Garnier_1}, we obtain a relation among
\begin{equation}
 \left\{
 p_1^{(l,m,0)},p_1^{(l,m+1,0)},
 p_2^{(l,m,0)},p_2^{(l,m+1,0)},
 q_1^{(l,m,0)},q_1^{(l,m+1,0)},
 q_2^{(l,m,0)},q_2^{(l,m+1,0)}
 \right\}.
\end{equation}
Using Equation \eqref{eqn:ODE_T1T2_Garnier_2}, we can show that the final relation obtained above indeed holds.
Hence, \eqref{eqn:lsG_Garniersol_1} is a special solution of Equation \eqref{eqn:auto_lsG}.

The remaining two solutions in (i), as well as (ii), can be verified by direct computation in the same manner.
This completes the proof.
\end{proof}

\section{Derivation of the transformations $T_i$ in \S \ref{section:main} from the symmetry groups}
\label{section:Backlund_P3P6Garnier}
In this section, we explain how the transformations $T_i$ used in \S \ref{section:main} are obtained from the symmetry groups (the groups formed by B\"acklund transformations) of {\rm P}$_{\rm III}$, {\rm P}$_{\rm VI}$, and {\rm Garnier}$_{\rm 2}$.
\subsection{The third Painlev\'e equation}\label{subsection:Weyl_P3}
In this subsection, we present the symmetry group of {\rm P}$_{\rm III}$ \cite{KNY2017:MR3609039,OkamotoK1987:MR927186,SakaiH2001:MR1882403}, and show that the transformations $T_1$ and $T_2$ given in \S \ref{subsection:main_P3} can be derived from this symmetry group.

Let $t\in\bbC$ be the independent variable, $p=p(t)\in\bbC$ and $q=q(t)\in\bbC$ dependent variables, and $a_1,a_2\in\bbC$ parameters.
We define the transformation group
\begin{equation}
 \widetilde{W}(2A_1^{(1)})=\langle s_1,s_2,\pi_1,\pi_2\rangle
\end{equation}
by the actions on the parameters $\{a_1,a_2\}$:
\begin{align*}
 &s_1:(a_1,a_2)\mapsto(-a_1,a_2),
 &&s_2:(a_1,a_2)\mapsto(a_1,-a_2),\\
 &\pi_1:(a_1,a_2)\mapsto(a_2,a_1),
 &&\pi_2:(a_1,a_2)\mapsto(1-a_1,a_2),
\end{align*}
those on the independent variable $t$:
\begin{equation*}
 s_1(t)=t,\quad
 s_2(t)=t,\quad
 \pi_1(t)=-t,\quad
 \pi_2(t)=t,
\end{equation*}
and those on the variables $\{p,q\}$:
\begin{align*}
 &s_1(p)=p,\quad
 s_1(q)=q+\dfrac{a_1}{p-1},\quad
 s_2(p)=p,\quad
 s_2(q)=q+\dfrac{a_2}{p},\\
 &\pi_1(p)=1-p,\quad
 \pi_1(q)=-q,\quad
 \pi_2(p)=-\dfrac{q(pq+a_2)}{t},\quad
 \pi_2(q)=\dfrac{t}{q}.
\end{align*}
As can be verified by direct computation, with respect to the action on the variables $\{p,q\}$, the transformations $\{s_1,s_2,\pi_2\}$ commute with the differential operator $d/dt$ of {\rm P}$_{\rm III}$ \eqref{eqn:Hamiltonian_system_P3}.
Moreover, since $\pi_1(t)=-t$, the transformation $\pi_1$ satisfies the following relation:
\begin{equation}
 \pi_1\,\dfrac{d}{dt}=-\dfrac{d}{dt}\,\pi_1=\dfrac{d}{d\pi_1(t)}\,\pi_1.
\end{equation}
Therefore, each element of $\widetilde{W}(2A_1^{(1)})$ is a B\"acklund transfomration of {\rm P}$_{\rm III}$ \eqref{eqn:Hamiltonian_system_P3}.
The elements of $\widetilde{W}(2A_1^{(1)})$ satisfy the following fundamental relations:
\begin{subequations}
\begin{align}
 &{s_1}^2={s_2}^2=1,\quad
 s_1s_2=s_2s_1,\\
 &{\pi_1}^2={\pi_2}^2=1,\quad
 (\pi_1\pi_2)^4=1,\\
 &\pi_1 s_1=s_2\pi_1,\quad
 \pi_1 s_2=s_1\pi_1,\quad
 (\pi_2s_1)^\infty=1,\quad
 \pi_2 s_2=s_2 \pi_2.
\end{align}
\end{subequations}
The transformations $T_1$ and $T_2$ in \S \ref{subsection:main_P3} are given as elements of $\widetilde{W}(2A_1^{(1)})$ by
\begin{equation}
 T_1=s_1\pi_2,\quad
 T_2=\pi_1s_1\pi_2\pi_1.
\end{equation}
As can also be seen from their actions on the parameters in \eqref{eqns:action_T1T2_P3}, the transformations $T_1$ and $T_2$ are translations in $\widetilde{W}(2A_1^{(1)})$.

\begin{remark}
Define the transformations $\{w^{(1)}_0,w^{(1)}_1,w^{(2)}_0,w^{(2)}_1,\pi^{(12)}\}$ by
\begin{equation}
 w^{(1)}_0=\pi_2,\quad
 w^{(1)}_1=s_1,\quad
 w^{(2)}_0=\pi_1\pi_2\pi_1,\quad
 w^{(2)}_1=s_2,\quad
 \pi^{(12)}=\pi_1.
\end{equation}
Then, the following fundamental relations hold:
\begin{subequations}
\begin{align}
 &\Big(w^{(i)}_0\Big)^2=\Big(w^{(i)}_1\Big)^2=\Big(w^{(i)}_0w^{(i)}_1\Big)^\infty=1,\quad
 i=1,2,\quad
 w^{(1)}_j w^{(2)}_k=w^{(2)}_k w^{(1)}_j,\quad 
 j,k=0,1,\\
 &\Big(\pi^{(12)}\Big)^2=1,\quad
 \pi^{(12)}w^{(1)}_0=w^{(2)}_0\pi^{(12)},\quad
 \pi^{(12)}w^{(1)}_1=w^{(2)}_1\pi^{(12)}.
\end{align}
\end{subequations}
That is, the transformation groups $\langle w^{(1)}_0,w^{(1)}_1\rangle$ and $\langle w^{(2)}_0,w^{(2)}_1\rangle$ each form an affine Weyl group of type $A_1^{(1)}$, and the two groups commute with each other.
Moreover, $\pi^{(12)}$ is a diagram automorphism that exchanges these two extended affine Weyl groups.
(See Figure \ref{fig:diagram_2A1}.)
\end{remark}

\begin{figure}[htbp]
\begin{center}
 \includegraphics[width=0.25\textwidth]{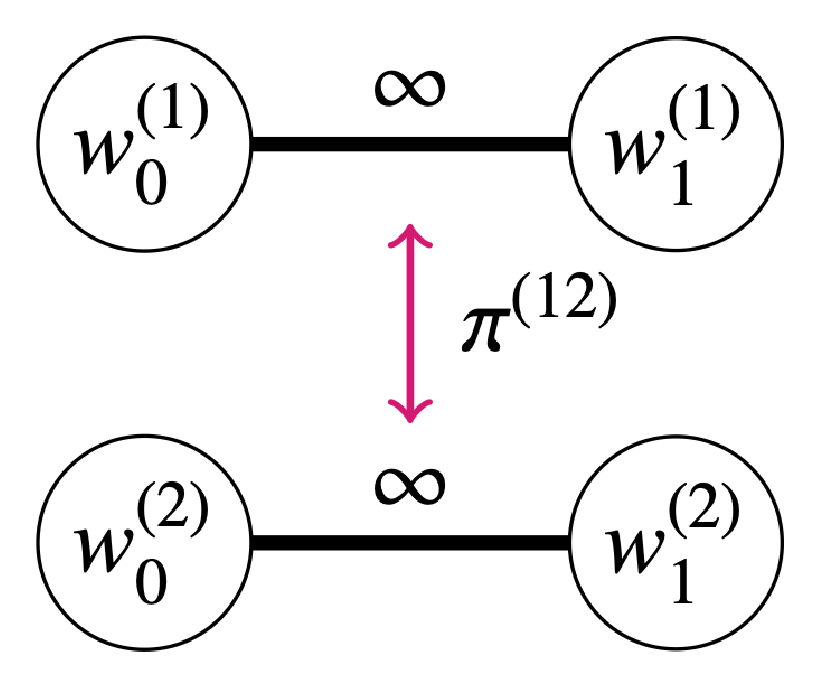}
\end{center}
\caption{\raggedright Coxeter diagram describing the relations among the transformations $\{w^{(1)}_0,w^{(1)}_1,w^{(2)}_0,w^{(2)}_1,\pi^{(12)}\}$.}
\label{fig:diagram_2A1}
\end{figure}

\subsection{The sixth Painlev\'e equation}\label{subsection:Weyl_P6}
Since the details of the discussion are the same as those in \S \ref{subsection:Weyl_P3}, we describe only the necessary information here in a concise manner.
For references on the symmetry group of {\rm P}$_{\rm VI}$, see for example \cite{KNY2017:MR3609039,OkamotoK1987:MR916698,SakaiH2001:MR1882403,JKMNS2017:MR3741826}.

\begin{description}
\item[The symmetry group of {\rm P}$_{\rm VI}$]
\begin{equation}
 \widetilde{W}(D_4^{(1)})=\langle s_0,s_1,s_2,s_3,s_4\rangle\rtimes\langle\si_1,\si_2,\si_3\rangle.
\end{equation}
\item[The actions on the parameters $\{a_0,\dots,a_4\}$]
\begin{align*}
 &s_i(a_i)=-a_i,\quad
 s_i(a_2)=a_2+a_i,\quad
 (i=0,1,3,4),\\
 &s_2(a_j)=a_j+a_2,\quad
 s_2(a_2)=-a_2,\quad
 (j=0,1,3,4),\\ 
 &\si_1:a_0\leftrightarrow a_1,~a_3\leftrightarrow a_4,\quad
 \si_2:a_0\leftrightarrow a_3,~a_1\leftrightarrow a_4,\quad
 \si_3:a_0\leftrightarrow a_4,~a_1\leftrightarrow a_3.
\end{align*}
\item[The actions on the independent variable $t$]
\begin{equation*}
 w(t)=t,\quad w\in\widetilde{W}(D_4^{(1)}).
\end{equation*}
\item[The actions on the variables $\{p,q\}$]
\begin{align*}
 &s_0(p)=p-\dfrac{a_0}{q-t},\quad
 s_2(q)=q+\dfrac{a_2}{p},\quad
 s_3(p)=p-\dfrac{a_3}{q-1},\quad
 s_4(p)=p-\dfrac{a_4}{q},\\
 &\si_1(p)=-\dfrac{(q-t)(pq-tp+a_2)}{t(t-1)},\quad
 \si_1(q)=\dfrac{t(q-1)}{q-t},\\
 &\si_2(p)=-\dfrac{q(pq+a_2)}{t},\quad
 \si_2(q)=\dfrac{t}{q},\quad
 \si_3(p)=\dfrac{(q-1)(pq-p+a_2)}{t-1},\quad
 \si_3(q)=\dfrac{q-t}{q-1}.
\end{align*}
\item[The fundamental relations]
\begin{subequations}
\begin{align}
 &{s_0}^2={s_1}^2={s_2}^2={s_3}^2={s_4}^2=1,\quad
 s_is_j=s_js_i,\quad (s_ks_2)^3=1,\quad i,j,k=0,1,3,4,\\
 &{\si_1}^2={\si_2}^2={\si_3}^2=\si_1\si_2\si_3=1,\quad
 \si_i\si_j=\si_j\si_i,\quad i,j=1,2,3,\\
 &\si_1s_0=s_1\si_1,\quad
 \si_1s_1=s_0\si_1,\quad
 \si_1s_2=s_2\si_1,\quad
 \si_1s_3=s_4\si_1,\quad
 \si_1s_4=s_3\si_1,\\
 &\si_2s_0=s_3\si_2,\quad
 \si_2s_1=s_4\si_2,\quad
 \si_2s_2=s_2\si_2,\quad
 \si_2s_3=s_0\si_2,\quad
 \si_2s_4=s_1\si_2,\\
 &\si_3s_0=s_4\si_3,\quad
 \si_3s_1=s_3\si_3,\quad
 \si_3s_2=s_2\si_3,\quad
 \si_3s_3=s_1\si_3,\quad
 \si_3s_4=s_0\si_3.
\end{align}
\end{subequations}
\item[The transformations $T_i$ in \S \ref{subsection:main_P6}]
\begin{equation}
 T_1=\si_1s_4s_3s_1s_0s_2,\quad
 T_2=\si_3s_1s_4s_2s_1s_4,\quad
 T_3=\si_2s_1s_3s_2s_1s_3.
\end{equation}
\end{description}

\begin{remark}\label{remark:Backlund_P6}
~\\[-1em]
\begin{itemize}
\item 
We follow the convention that parameters and variables not explicitly included in the actions listed in the equations above remain unchanged under the corresponding transformation.  
That is, the transformation acts as the identity on those parameters and variables.
\item
Under the action on the variables $\{p,q\}$, the elements of the group $\widetilde{W}(D_4^{(1)})$ commute with the differential operator $d/dt$ of {\rm P}$_{\rm VI}$ \eqref{eqn:Hamiltonian_system_P6}.
\item 
The group $\widetilde{W}(D_4^{(1)})$ forms an extended affine Weyl group of type $D_4^{(1)}$. 
\item
As can also be seen from the action on the parameters \eqref{eqns:action_T1T2T3_P6}, the transformations $T_i$, $i=1,2,3$, are not translations of $\widetilde{W}(D_4^{(1)})$, but ${T_i}^2$ are translations.
\item
The transformations $T_1$, $T_2$, $T_3$ in \S \ref{subsection:main_P6} are equal to the transformations $\rho_1,\rho_2,\rho_3$ in \cite{JKMNS2017:MR3741826}, respectively.
\end{itemize}
\end{remark}

\subsection{The Garnier system in two variables}\label{subsection:Weyl_Garnier}
Since the details of the discussion are the same as those in \S \ref{subsection:Weyl_P3}, we describe only the necessary information here in a concise manner.
For references on the symmetry group of {\rm Garnier}$_{\rm 2}$, see for example \cite{JKMN2021:zbMATH07653201,KimuraH1990:MR1078402,KO1984:MR776915,SuzukiT2005:MR2177118,TsudaT2003:MR1998470,TsudaT2003:MR1987136}.

\begin{description}
\item[The symmetry group of {\rm Garnier}$_{\rm 2}$]
\begin{equation}\label{eqn:transformationgroup_Garnier}
 {\mathcal G}_{\rm Garnier}=\langle r_1,r_2,r_3,r_4,r_5,r_{34},\si_{12},\si_{23},\si_{34},\si_{45}\rangle.
\end{equation}
\item[The actions on the parameters $\{\theta_1,\theta_2,\ka_0,\ka_1,\ka_\infty\}$]
\begin{align*}
 &r_1(\theta_1)=-\theta_1,\quad
 r_2(\theta_2)=-\theta_2,\quad
 r_3(\ka_0)=-\ka_0,\quad
 r_4(\ka_1)=-\ka_1,\quad
 r_5(\ka_\infty)=-\ka_\infty,\\
 &r_{34}:(\theta_1,\theta_2,\ka_0,\ka_1,\ka_\infty)
 \mapsto(-\theta_1,-\theta_2,1-\ka_0,1-\ka_1,-\ka_\infty),\\
 &\si_{12}:\theta_1\leftrightarrow\theta_2,\quad
 \si_{23}:\theta_2\leftrightarrow\ka_0,\quad
 \si_{34}:\ka_0\leftrightarrow\ka_1,\quad
 \si_{45}:\ka_1\leftrightarrow\ka_\infty.
\end{align*}
\item[The actions on the independent variables $\{t_1,t_2\}$]
\begin{align*}
 &\si_{12}:(t_1,t_2)
 \mapsto(t_2,t_1),
 &&\si_{23}:(t_1,t_2)
 \mapsto\left(\dfrac{t_2-t_1}{t_2-1},\dfrac{t_2}{t_2-1}\right),\\
 &\si_{34}:(t_1,t_2)
 \mapsto\left(\dfrac{1}{t_1},\dfrac{1}{t_2}\right),
 &&\si_{45}:(t_1,t_2)
 \mapsto\left(\dfrac{t_1}{t_1-1},\dfrac{t_2}{t_2-1}\right).
\end{align*}
\item[The actions on the variables $\{p_1,p_2,q_1,q_2\}$]
\begin{align*}
 &r_1(p_1)=p_1-\dfrac{\theta_1}{q_1},\qquad
 r_2(p_2)=p_2-\dfrac{\theta_2}{q_2},\\
 &r_3:(p_1,p_2)\mapsto\left(p_1-\dfrac{\ka_0}{t_1(g_t-1)},p_2-\dfrac{\ka_0}{t_2(g_t-1)}\right),\\
 &r_4:(p_1,p_2)\mapsto\left(p_1-\dfrac{\ka_1}{g_1-1},p_2-\dfrac{\ka_1}{g_1-1}\right),\\
 &r_{34}(p_1)
 =-\dfrac{p_1q_1(p_1q_1+p_2q_2+\ga)(p_1q_1+p_2q_2+\ga+\ka_\infty)}{t_1p_1 (p_1q_1-\theta_1)},\\
 &r_{34}(p_2)
 =-\dfrac{p_2q_2(p_1q_1+p_2q_2+\ga)(p_1q_1+p_2q_2+\ga+\ka_\infty)}{t_2p_2(p_2q_2-\theta_2)},\\
 &r_{34}(q_1)
 =\dfrac{t_1p_1(p_1q_1-\theta_1)}{(p_1q_1+p_2q_2+\ga)(p_1q_1+p_2q_2+\ga+\ka_\infty)},\\
 &r_{34}(q_2)
 =\dfrac{t_2p_2(p_2q_2-\theta_2)}{(p_1q_1+p_2q_2+\ga)(p_1q_1+p_2q_2+\ga+\ka_\infty)},\\
 &\si_{12}:(p_1,p_2,q_1,q_2)
 \mapsto(p_2,p_1,q_2,q_1),\\
 &\si_{23}:(p_1,p_2,q_1,q_2)
 \mapsto\left(\dfrac{(t_2-1)(t_1p_1-t_2p_2)}{t_2-t_1},(1-t_2)p_2,\dfrac{(t_2-t_1)q_1}{t_1(t_2-1)},\dfrac{t_2(g_t-1)}{1-t_2}\right),\\
 &\si_{34}:(p_1,p_2,q_1,q_2)
 \mapsto\left(t_1p_1,t_2p_2,\dfrac{q_1}{t_1},\dfrac{q_2}{t_2}\right),\\
 &\si_{45}(p_1)=(g_1-1)(p_1-\ga-p_1q_1-p_2q_2),\\
 &\si_{45}(p_2)=(g_1-1)(p_2-\ga-p_1q_1-p_2q_2),\\
 &\si_{45}(q_1)=\dfrac{q_1}{g_1-1},\quad
 \si_{45}(q_2)=\dfrac{q_2}{g_1-1},
\end{align*}
where $\ga$ is given by Equation \eqref{eqn:def_gamma}, and $g_t$ and $g_1$ are given by
\begin{equation}
 g_t=\dfrac{q_1}{t_1}+\dfrac{q_2}{t_2},\quad
 g_1=q_1+q_2. 
\end{equation}
\item[The fundamental relations]
\begin{subequations}
\begin{align}
 &{r_1}^2={r_2}^2={r_3}^2={r_4}^2={r_5}^2=1,\quad 
 r_ir_j=r_jr_i\quad(i\ne j),\\
 &{r_{34}}^2=1,\quad
 r_ir_{34}=r_{34}r_i\quad(i=1,2,5),\quad
 (r_3r_{34})^\infty=(r_4r_{34})^\infty=1,\\
 &{\si_{12}}^2={\si_{23}}^2={\si_{34}}^2={\si_{45}}^2=1,\\
 &\si_{12}r_1=r_2\si_{12},\quad
 \si_{12}r_2=r_1\si_{12},\quad
 \si_{12}r_i=r_i\si_{12},\quad (i=3,4,5),\\
 &\si_{23}r_2=r_3\si_{32},\quad
 \si_{23}r_3=r_2\si_{23},\quad
 \si_{23}r_i=r_i\si_{23},\quad (i=1,4,5),\\
 &\si_{34}r_3=r_4\si_{34},\quad
 \si_{34}r_4=r_3\si_{34},\quad
 \si_{34}r_i=r_i\si_{34},\quad (i=1,2,5),\\
 &\si_{45}r_4=r_5\si_{45},\quad
 \si_{45}r_5=r_4\si_{45},\quad
 \si_{45}r_i=r_i\si_{45},\quad (i=1,2,3),\\
 &\si_{12}r_{34}=r_{34}\si_{12},\quad
 \si_{34}r_{34}=r_{34}\si_{34},\quad
 (\si_{23}r_{34})^\infty=(\si_{45}r_{34})^\infty=1,\\
 &\si_{12}\si_{34}=\si_{34}\si_{12},\quad
 \si_{12}\si_{45}=\si_{45}\si_{12},\quad
 \si_{23}\si_{45}=\si_{45}\si_{23},\\
 &(\si_{12}\si_{23})^3=(\si_{23}\si_{34})^3=(\si_{34}\si_{45})^3=1.
\end{align}
\end{subequations}
\item[The transformations $T_i$ in \S \ref{subsection:main_Garnier}]
\begin{equation}
 T_1=r_1r_2r_3r_5r_{14},\quad
 T_2=r_1r_2r_3r_5r_{24},\quad
 T_3=r_1r_2r_3r_5r_{34},
\end{equation}
where
\begin{equation}
 r_{14}=\si_{12}r_{24}\si_{12},\quad
 r_{24}=\si_{23}r_{34}\si_{23}.
\end{equation}
\end{description}

\begin{remark}~\\[-1em]
\begin{itemize}
\item
See the first item of Remark \ref{remark:Backlund_P6} for the convention on how to write the actions.
\item
When acting on the variables $\{p_1,p_2,q_1,q_2\}$, the transformations $\{r_1,\dots,r_5,r_{34}\}$ commute with the differential operators $\{\partial_{t_1},\partial_{t_2}\}$ of {\rm Garnier}$_{\rm 2}$ \eqref{eqn:Hamiltonian_system_Garnier}.
Moreover, for $w\in\{\si_{12},\si_{23},\si_{34},\si_{45}\}$ and $i=1,2$, we have
\begin{equation}
 w\partial_{t_i}=\partial_{w(t_i)}w.
\end{equation}
More explicitly, the relations are given by
\begin{subequations}
\begin{align}
 &\si_{12}\partial_{t_1}=\partial_{t_2}\si_{12}=\partial_{\si_{12}(t_1)}\si_{12},\\
 &\si_{12}\partial_{t_2}=\partial_{t_1}\si_{12}=\partial_{\si_{12}(t_2)}\si_{12},\\
 &\si_{23}\partial_{t_1}
 =\dfrac{\partial t_1}{\partial\si_{23}(t_1)}\,\partial_{t_1}\si_{23}
 =\left(\dfrac{\partial t_1}{\partial\si_{23}(t_1)}\,\partial_{t_1}+\dfrac{\partial t_2}{\partial\si_{23}(t_1)}\,\partial_{t_2}\right)\si_{23}
 =\partial_{\si_{23}(t_1)}\si_{23},\\
 &\si_{23}\partial_{t_2}
 =\left(\dfrac{\partial t_1}{\partial\si_{23}(t_2)}\,\partial_{t_1}+\dfrac{\partial t_2}{\partial\si_{23}(t_2)}\,\partial_{t_2}\right)\si_{23}
 =\partial_{\si_{23}(t_2)}\si_{23},\\
 &\si_{34}\partial_{t_1}
 =\dfrac{dt_1}{d\si_{34}(t_1)}\,\partial_{t_1}\si_{34}
 =\partial_{\si_{34}(t_1)}\si_{34},\\
 &\si_{34}\partial_{t_2}
 =\dfrac{dt_2}{d\si_{34}(t_2)}\,\partial_{t_2}\si_{34}
 =\partial_{\si_{34}(t_2)}\si_{34},\\
 &\si_{45}\partial_{t_1}
 =\dfrac{dt_1}{d\si_{45}(t_1)}\,\partial_{t_1}\si_{45}
 =\partial_{\si_{45}(t_1)}\si_{45},\\
 &\si_{45}\partial_{t_2}
 =\dfrac{dt_2}{d\si_{45}(t_2)}\,\partial_{t_2}\si_{45}
 =\partial_{\si_{45}(t_2)}\si_{45}.
\end{align}
\end{subequations}
\item
The group ${\mathcal G}_{\rm Garnier}$ includes the affine Weyl group of type $B_5^{(1)}$ \cite{SuzukiT2005:MR2177118}. 
We note that the full symmetry group of {\rm Garnier}$_{\rm 2}$ has not yet been completely identified. 
\item
The transformations $T_1$, $T_2$, $T_3$ in \S \ref{subsection:main_Garnier} are equal to the transformations ${\rho_1}^{-1}$, ${\rho_2}^{-1}$, ${\rho_3}^{-1}$ in \cite{JKMN2021:zbMATH07653201}, respectively.
\end{itemize}
\end{remark}

\section{Concluding remarks}\label{ConcludingRemarks}
In this paper, we have shown that Equations \eqref{eqn:auto_dkdv}--\eqref{eqn:auto_dV} admit special solutions via {\rm P}$_{\rm III}$, {\rm P}$_{\rm VI}$, or {\rm Garnier}$_{\rm 2}$, depending on the equation.
It is noteworthy that although Equations \eqref{eqn:auto_dkdv}--\eqref{eqn:auto_dV} are autonomous, they possess special solutions expressed in terms of non-autonomous {\ODE}s, which can be interpreted as B\"acklund transformations of the corresponding non-autonomous differential equations.

The novelty of this work lies in the fact that special solutions of autonomous integrable {\PDE}s can be expressed in terms of the dependent variables of Painlev\'e-type non-autonomous {\ODE}s.
This phenomenon points to an aspect of autonomous integrable {\PDE}s that is not yet fully understood, namely, why their special solutions can be described by Painlev\'e-type dynamics.
It should be emphasized that the special solutions obtained in this paper were found in a non-systematic manner, and that a comprehensive construction method has not yet been established.
If such a systematic method were to be developed, it would be expected to reveal new connections between integrable {\PDE}s and Painlev\'e-type dynamics.

In the paper \cite{nakazono2026variationdKdV}, it was shown that special solutions of a variation of the dKdV equation:
\begin{equation}\label{eqn:variationdKdV}
 \ep\, u_{l+1,m+1}-u_{l,m}=\dfrac{\ep}{u_{l,m+1}}-\dfrac{~1~}{u_{l+1,m}},
\end{equation}
where $\ep\in\bbC^\ast$ is an additional parameter,
are given by solutions of $q$-Painlev\'e equations, which are multiplicative-type discrete Painlev\'e equations (see, {\it e.g.}, \cite{SakaiH2001:MR1882403,KNY2017:MR3609039} for $q$-Painlev\'e equations).
However, since the bilinear difference equation satisfied by the $\tau$ function associated with Equation \eqref{eqn:variationdKdV} is non-autonomous (see Remark 1.1 in \cite{nakazono2026variationdKdV}), it is understandable that special solutions expressed in such non-autonomous {\ODE}s can exist.
On the other hand, the autonomous {\PDE}s \eqref{eqn:auto_dkdv}--\eqref{eqn:auto_dV} are autonomous also from the viewpoint of $\tau$ functions, and nevertheless possess special solutions expressed by non-autonomous {\ODE}s obtained by B\"acklund transformations of Painlev\'e equations and the Garnier system.
In this sense, they provide more mysterious examples.
However, similar phenomena have already been reported in several works, where special solutions are regarded not as the dependent variables of Painlev\'e equations or the Garnier system themselves, but as their $\tau$ functions.
For example, see the work by Nijhoff {\it et al.} \cite{NRGO2001:MR1819383} and its extension by Joshi {\it et al.} \cite{JKMN2021:zbMATH07653201}.
In the former, a special solution of the autonomous version of the lattice Schwarzian KdV equation \cite{NC1995:MR1329559} is constructed by the $\tau$ function of {\rm P}$_{\rm VI}$, while in the latter, it is constructed by the $\tau$ function of {\rm Garnier}$_{\rm 2}$.

Before concluding this section, let us briefly explain our future work.
In the paper \cite{nakazono2026variationdKdV}, it was shown that the autonomous {\PDE} \eqref{eqn:variationdKdV} admits discrete Painlev\'e transcendent solutions (dP solutions) given by $q$-Painlev\'e equations.
Moreover, in the paper \cite{nakazono2022discrete}, it was shown that non-autonomous {\PDE}s admit dP solutions given by $q$-Painlev\'e equations.
Thus, dP solutions obtained from $q$-Painlev\'e equations have been reported for both autonomous and non-autonomous {\PDE}s.
What, then,  about cases other than $q$-Painlev\'e equations?
As mentioned in the Introduction, noting that the special solutions given in this paper can also be called dP solutions, we have shown that {\rm P}$_{\rm III}$, {\rm P}$_{\rm VI}$, and {\rm Garnier}$_{\rm 2}$ provide dP solutions of the autonomous {\PDE}s \eqref{eqn:auto_dkdv}--\eqref{eqn:auto_dV} (see Theorems \ref{theorem:P3solutions}, \ref{theorem:P6solutions} and \ref{theorem:Garniersolutions}).
Moreover, in Appendix \ref{appendix:additive_dKdV}, we show that {\rm P}$_{\rm III}$ also provides a dP solution of a non-autonomous {\PDE}.
Although this is not stated explicitly in \cite{NakazonoN2018:MR3760161}, it was shown that {\rm P}$_{\rm IV}$ and {\rm P}$_{\rm V}$ provide dP solutions of a non-autonomous {\PDE}.
Thus, even in cases other than $q$-Painlev\'e equations, it is possible to construct dP solutions not only for autonomous {\PDE}s but also for non-autonomous {\PDE}s.
While this paper mainly deals with autonomous {\PDE}s, in a forthcoming paper, we will show that dP solutions of various non-autonomous {\PDE}s can be constructed using {\rm P}$_{\rm IV}$ and {\rm P}$_{\rm V}$.
Finally, we note that in the work by Tsoubelis and Xenitidis \cite{TX2009:zbMATH05559172}, special solutions of non-autonomous {\PDE}s were constructed using {\rm P}$_{\rm V}$ and {\rm P}$_{\rm VI}$, and these results can be understood, similarly to those of Nijhoff {\it et al.} \cite{NRGO2001:MR1819383} and Joshi {\it et al.} \cite{JKMN2021:zbMATH07653201}, as being related to underlying $\tau$ function structures.

\subsection*{Acknowledgment}
The author thanks Dr. Dinh T. Tran for pointing out that the HV equation \eqref{eqn:auto_HV} appears in \cite{HV2007:Searching}.
The author is also grateful to Prof. Frank W. Nijhoff for suggesting that the HV equation \eqref{eqn:auto_HV} can be obtained as a limit of the $Q1_{\de=1}$ equation \eqref{eqn:auto_Q1}.
The author acknowledges the use of ChatGPT for proofreading the manuscript and assisting with translation into English.
This work was supported by JSPS KAKENHI, Grant Number JP23K03145.
\appendix
\section{A special solution of a non-autonomous dKdV equation}\label{appendix:additive_dKdV}
In this appendix, we show that the additive-type non-autonomous version of the dKdV equation:
\begin{equation}\label{eqn:additive_dkdv}
 u_{l+1,m+1}-u_{l,m}=\dfrac{{\be_{m+1}}^2-{\al_l}^2}{u_{l,m+1}}-\dfrac{{\be_m}^2-{\al_{l+1}}^2}{u_{l+1,m}},
\end{equation}
where
\begin{equation}
 \al_l=\al_0+l,\quad
 \be_m=\be_0+m,
\end{equation}
and $\al_0,\be_0\in\bbC$ are parameters, admits a special solution via {\rm P}$_{\rm III}$. 
The result is given in the following theorem.
Note that since it can be shown in the same way as in the case of Theorem \ref{theorem:P3solutions}, the proof of the following theorem is omitted.
\begin{theorem}\label{theorem:P3solution_2}
Equation \eqref{eqn:additive_dkdv} admits the following special solution:
\begin{equation}\label{eqn:additive_dkdv_P3sol}
 u_{l,m}=a_2^{(l-m)}\left(\dfrac{1}{p^{(l+m,l-m)}}-1\right),
\end{equation}
with
\begin{equation}
 \al_l=\dfrac{a_1^{(l)}+a_2^{(l)}}{2}\quad \text{and}\quad \be_m=\dfrac{a_1^{(m)}-a_2^{(-m)}}{2}.
\end{equation}
Here, $a_1^{(k_1)}$, $a_2^{(k_2)}$, and $p^{(k_1,k_2)}$ are defined in \S \ref{subsection:main_P3}.
\end{theorem}

\section{The CAC property}\label{appendix:CAC}
In this Appendix, we first explain the definition of the CAC property for {\PDE}s of a specific form, and then, in a subsection, we show that the HV equation \eqref{eqn:auto_HV} has the CAC property.
Note that the CAC property is originally a local property defined not for {\PDE}s themselves, but for relations satisfied by eight variables assigned to the vertices of a cube.
For details on the CAC property, see \cite{NQC1983:MR719638,NCWQ1984:MR763123,QNCL1984:MR761644,NS1998:zbMATH01844203,NW2001:MR1869690}.
In addition, we construct a Lax pair of the HV equation \eqref{eqn:auto_HV} using the CAC property.
See, for example, \cite{BS2002:MR1890049,NijhoffFW2002:MR1912127,WalkerAJ:thesis} for references on a method for constructing Lax pairs using the CAC property.

Let us consider the following system of \PDE s:
\begin{subequations}\label{eqns:CAC_PDE_ABCA'}
\begin{align}
 &A(u_{l,m},u_{l+1,m},u_{l,m+1},u_{l+1,m+1})=0,\label{eqn:CAC_PDE_A}\\
 &B(u_{l,m},u_{l,m+1},v_{l,m},v_{l,m+1})=0,\label{eqn:CAC_PDE_B}\\
 &C(u_{l,m},u_{l+1,m},v_{l,m},v_{l+1,m})=0,\label{eqn:CAC_PDE_C}\\
 &A'(v_{l,m},v_{l+1,m},v_{l,m+1},v_{l+1,m+1})=0.\label{eqn:CAC_PDE_AA}
\end{align}
\end{subequations}
Here, the functions $A$, $B$, $C$, and $A'$ are irreducible multilinear polynomials in four variables.
Consider the following sublattice of the lattice $\bbZ^3$:
\begin{equation}\label{eqn:lattice_01}
 \set{(l,m,0)\in\bbZ^3}{l,m\in\bbZ}\cup\set{(l,m,1)\in\bbZ^3}{l,m\in\bbZ}
\end{equation}
and assign the $u$- and $v$-variables on the vertices of the sublattice by the following correspondences:
\begin{equation}\label{eqn:lattice_uv}
 (l,m,0)\,\leftrightarrow\,u_{l,m},\qquad
 (l,m,1)\,\leftrightarrow\,v_{l,m}.
\end{equation}
Then, focusing on the cube given by the 8 points
\begin{align}
 &(l,m,0),\quad (l+1,m,0),\quad (l,m+1,0),\quad (l+1,m+1,0),\notag\\
 &(l,m,1),\quad (l+1,m,1),\quad (l,m+1,1),\quad (l+1,m+1,1),
 \label{eqn:lattice_8points}
\end{align}
from the system \eqref{eqns:CAC_PDE_ABCA'}, we obtain the following face-equations of the cube (see Figure \ref{fig:CACcube}):
\begin{subequations}\label{eqns:CAC_ABCA'}
\begin{align}
 &{\mathcal A}=A(u_{l,m},u_{l+1,m},u_{l,m+1},u_{l+1,m+1})=0,\label{eqn:CAC_A}\\
 &{\mathcal A}'=A'(v_{l,m},v_{l+1,m},v_{l,m+1},v_{l+1,m+1})=0,\label{eqn:CAC_AA}\\
 &{\mathcal B}=B(u_{l,m},u_{l,m+1},v_{l,m},v_{l,m+1})=0,\label{eqn:CAC_B}\\
 &\overline{{\mathcal B}}=B(u_{l+1,m},u_{l+1,m+1},v_{l+1,m},v_{l+1,m+1})=0,\label{eqn:CAC_BB}\\
 &{\mathcal C}=C(u_{l,m},u_{l+1,m},v_{l,m},v_{l+1,m})=0,\label{eqn:CAC_C}\\
 &\widetilde{{\mathcal C}}=C(u_{l,m+1},u_{l+1,m+1},v_{l,m+1},v_{l+1,m+1})=0.\label{eqn:CAC_CC}
\end{align}
\end{subequations}
The CAC property is defined as follows.
\begin{enumerate}
\item 
There are following three ways to calculate $v_{l+1,m+1}$ by using all the equations in the system \eqref{eqns:CAC_ABCA'} with $\{u_{l,m},u_{l+1,m},u_{l,m+1},v_{l,m}\}$ as initial values.
\begin{enumerate}
\item 
Express $v_{l+1,m+1}$ as a rational function in terms of $\{v_{l,m},v_{l+1,m},v_{l,m+1}\}$ by using Equation \eqref{eqn:CAC_AA}.
Then, eliminate $v_{l,m+1}$ by using Equation \eqref{eqn:CAC_B} and $v_{l+1,m}$ by using Equation \eqref{eqn:CAC_C} from it.\\[-0.7em]
\item 
Express $v_{l+1,m+1}$ as a rational function in terms of $\{u_{l+1,m},u_{l+1,m+1},v_{l+1,m}\}$ by using Equation \eqref{eqn:CAC_BB}.
Then, eliminate $u_{l+1,m+1}$ by using Equation \eqref{eqn:CAC_A} and $v_{l+1,m}$ by using Equation \eqref{eqn:CAC_C} from it.\\[-0.7em]
\item 
Express $v_{l+1,m+1}$ as a rational function in terms of $\{u_{l,m+1},u_{l+1,m+1},v_{l,m+1}\}$ by using Equation \eqref{eqn:CAC_CC}.
Then, eliminate $u_{l+1,m+1}$ by using Equation \eqref{eqn:CAC_A} and $v_{l,m+1}$ by using Equation \eqref{eqn:CAC_B} from it.\\[-0.7em]
\end{enumerate}
When $v_{l+1,m+1}$ is uniquely determined as a rational function with the initial values $\{v_{l,m},u_{l+1,m},u_{l,m+1},v_{l,m}\}$, then the system \eqref{eqns:CAC_PDE_ABCA'} is said to have the {\it CAC property} or said to be a {\it CAC-system}.
\item
We say that the {\PDE} \eqref{eqn:CAC_PDE_A} (or the {\PDE} \eqref{eqn:CAC_PDE_AA}) has the CAC property if the system \eqref{eqns:CAC_PDE_ABCA'} has the CAC property. 
The pair of equations \eqref{eqn:CAC_PDE_B} and \eqref{eqn:CAC_PDE_C} is then referred to as a {\it CAC-type BT} from the {\PDE} \eqref{eqn:CAC_PDE_A} to the {\PDE} \eqref{eqn:CAC_PDE_AA} (or from the {\PDE} \eqref{eqn:CAC_PDE_AA} to the {\PDE} \eqref{eqn:CAC_PDE_A}).
When the equations \eqref{eqn:CAC_PDE_A} and \eqref{eqn:CAC_PDE_AA} are the same {\PDE}, the CAC-type BT \eqref{eqn:CAC_PDE_B} and \eqref{eqn:CAC_PDE_C} is specifically called a {\it CAC-type auto-BT} of the {\PDE} \eqref{eqn:CAC_PDE_A}.
\end{enumerate}

\begin{figure}[htbp]
\begin{center}
 \includegraphics[width=0.5\textwidth]{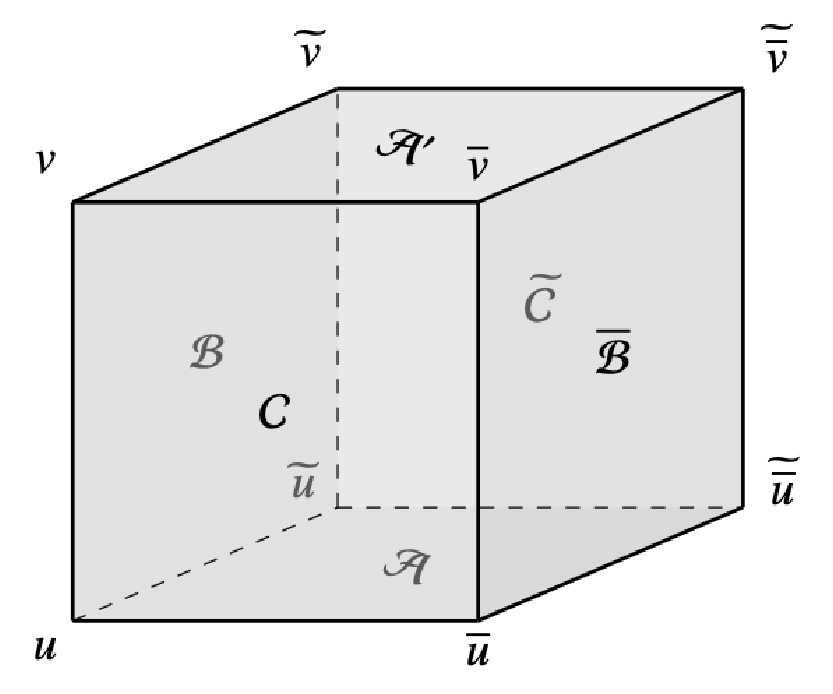}
\end{center}
\caption{A cube for the system \eqref{eqns:CAC_PDE_ABCA'}.
The $u$- and $v$-variables are assigned at the bottom and top vertices of the cube, respectively.
Also, the equations \eqref{eqns:CAC_ABCA'} are assigned to the faces of the cube.
For simplicity, we here use the shorthand notations:~
$u=u_{l,m}$,\quad
$v=v_{l,m}$,\quad
$\overline{\rule{0em}{0.3em}\hspace{0.6em}}:\,l\to l+1$,\quad
$\widetilde{}\,:\,m\to m+1$.}
\label{fig:CACcube}
\end{figure}

\subsection{The CAC property of the HV equation}
The system of {\PDE}s:
\begin{subequations}\label{eqns:HV_CACautoBT}
\begin{align}
 &(u_{l,m}-v_{l,m})(u_{l,m+1}-v_{l,m+1})-\la(u_{l,m}-u_{l,m+1})(v_{l,m}-v_{l,m+1})-\la(u_{l,m}-v_{l,m+1})=0,\label{eqn:HV_CACautoBT_1}\\
 &(u_{l,m}-v_{l,m})(u_{l+1,m}-v_{l+1,m})-\la(u_{l,m}-v_{l+1,m})=0,\label{eqn:HV_CACautoBT_2}
\end{align}
\end{subequations}
where $\la\in\bbC$,
is a CAC-type auto-BT of the HV equation \eqref{eqn:auto_HV}.
That is, the $v$-variable satisfies the following {\PDE}:
\begin{equation}
 (v_{l,m}-v_{l,m+1})(v_{l+1,m}-v_{l+1,m+1})=v_{l,m+1}-v_{l+1,m}.
\end{equation}

\begin{remark}
This CAC system corresponds to Equation (3.9)$_{\ep=0}$ in \cite{BollR:thesis}.
Therefore, in the classification of polynomials by Adler-Bobenko-Suris \cite{ABS2003:MR1962121,ABS2009:MR2503862} and Boll \cite{BollR2012:MR3010833,BollR2011:MR2846098,BollR:thesis}, the HV equation \eqref{eqn:auto_HV} and Equation \eqref{eqn:HV_CACautoBT_2} are of $H2_{\ep=0}$-type, and Equation \eqref{eqn:HV_CACautoBT_1} is of $Q1_{\de=1}$-type.
Moreover, it can be verified by direct calculation that this CAC system is strong in the sense of Hietarinta \cite{HietarintaJ2019:zbMATH07053246}.
\end{remark}

Then, a Lax pair of the the HV equation \eqref{eqn:auto_HV} is given by
\begin{equation}\label{eqn:Lax_phi}
 \phi_{l+1,m}=L_{l,m}\phi_{l,m},\quad
 \phi_{l,m+1}=M_{l,m}\phi_{l,m},
\end{equation}
where
\begin{subequations}
\begin{align}
 &L_{l,m}=\dfrac{1}{u_{l,m}-u_{l+1,m}}
 \begin{pmatrix}
  u_{l+1,m}&u_{l,m}(\la-u_{l+1,m})\\
  1&\la-u_{l,m}
 \end{pmatrix},\\
 &M_{l,m}=\dfrac{1}{u_{l,m}-u_{l,m+1}}
 \begin{pmatrix}
  \la u_{l,m}-(\la-1)u_{l,m+1}&u_{l,m}(\la-u_{l,m+1})\\
  1&(\la-1)u_{l,m}-\la u_{l,m+1}+\la
 \end{pmatrix}.
\end{align}
\end{subequations}
Indeed, the compatibility condition
\begin{equation}\label{eqn:compatibility_condition}
 L_{l,m+1}M_{l,m}=M_{l+1,m}L_{l,m}
\end{equation}
gives the HV equation \eqref{eqn:auto_HV}.
\section{The CABC property}\label{appendix:CABC}
In this Appendix, we first explain the definition of the CABC property for {\PDE}s of a specific form, and then, in subsections, we show that the dVolterra equation \eqref{eqn:auto_dV} has the CABC property.
Note that the CABC property is originally a local property defined not for {\PDE}s themselves, but for relations satisfied by eight variables assigned to the vertices of a cube.
For details on the CABC property, see \cite{JN2021:zbMATH07476241}.
In addition, we construct two Lax pairs of the dVolterra equation \eqref{eqn:auto_dV} using the CABC property.
A method for constructing Lax pairs using the CABC property is also given in \cite{JN2021:zbMATH07476241}.

Let us consider the following system of \PDE s:
\begin{subequations}\label{eqns:CABC_PDE_ASBC}
\begin{align}
 &A(u_{l,m},u_{l+1,m},u_{l,m+1},u_{l+1,m+1})=0,\label{eqn:CABC_PDE_A}\\
 &S(u_{l,m},u_{l+1,m},v_{l,m+1},v_{l+1,m+1})=0,\label{eqn:CABC_PDE_S}\\
 &B(u_{l,m},v_{l,m},v_{l,m+1})=0,\label{eqn:CABC_PDE_B}\\
 &C(u_{l,m},u_{l+1,m},v_{l,m},v_{l+1,m})=0.\label{eqn:CABC_PDE_C}
\end{align}
\end{subequations}
Here, the functions $A$, $S$, and $C$ are irreducible multilinear polynomials in four variables.
The function $B=B(x,y,z)$ is a polynomial in three variables satisfying the following:
\begin{enumerate}
\item[1)]
$\deg_x B\geq 1$,\quad
$\deg_y B=\deg_z B=1$;
\item[2)]
Let $y=f(x,z)$ be the solution of $B=0$.
Then, $f(x,z)$ is a rational function that depends on $x$ and $z$, that is, 
the following hold:
\begin{equation}
 \dfrac{\partial}{\partial x} f(x,z)\neq0,\quad
 \dfrac{\partial}{\partial z} f(x,z)\neq0.
\end{equation}
\end{enumerate}
Assign the $u$- and $v$-variables on the sublattice \eqref{eqn:lattice_01} by the correspondence \eqref{eqn:lattice_uv}.
Then, by considering the cube consisting of the eight vertices \eqref{eqn:lattice_8points} (see Figure \ref{fig:CABCcube}), the system of equations around the cube obtained from the system \eqref{eqns:CABC_PDE_ASBC} is the following:
\begin{subequations}\label{eqns:CABC_ASBC}
\begin{align}
 &{\mathcal A}=A(u_{l,m},u_{l+1,m},u_{l,m+1},u_{l+1,m+1})=0,\label{eqn:CABC_A}\\
 &{\mathcal S}=S(u_{l,m},u_{l+1,m},v_{l,m+1},v_{l+1,m+1})=0,\label{eqn:CABC_S}\\
 &{\mathcal B}=B(u_{l,m},v_{l,m},v_{l,m+1})=0,\label{eqn:CABC_B}\\
 &\overline{{\mathcal B}}=B(u_{l+1,m},v_{l+1,m},v_{l+1,m+1})=0,\label{eqn:CABC_BB}\\
 &{\mathcal C}=C(u_{l,m},u_{l+1,m},v_{l,m},v_{l+1,m})=0,\label{eqn:CABC_C}\\
 &\widetilde{{\mathcal C}}=C(u_{l,m+1},u_{l+1,m+1},v_{l,m+1},v_{l+1,m+1})=0.\label{eqn:CABC_CC}
\end{align}
\end{subequations}

\begin{figure}[htbp]
\begin{center}
 \includegraphics[width=0.5\textwidth]{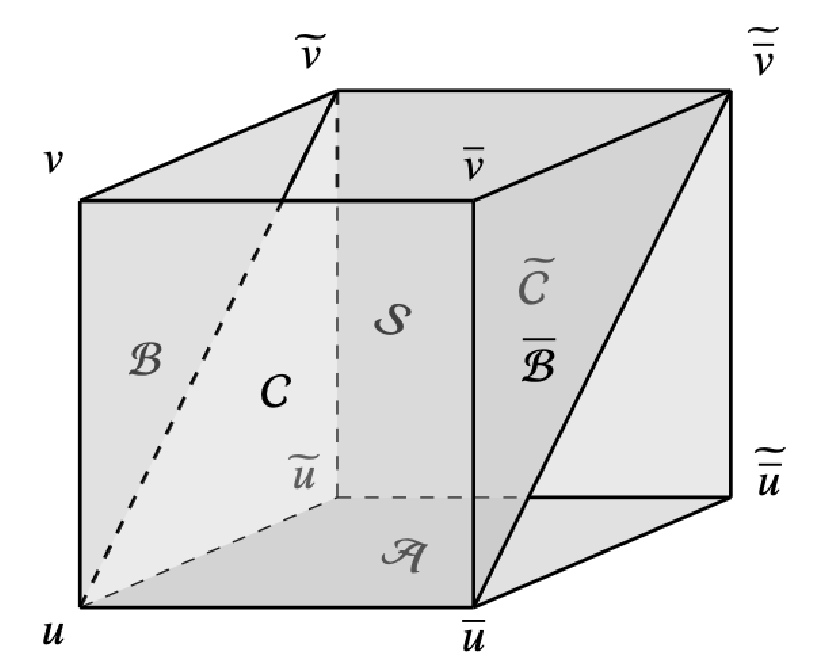}
\end{center}
\caption{A cube for the system \eqref{eqns:CABC_PDE_ASBC}.
The $u$- and $v$-variables are assigned at the bottom and top vertices of the cube, respectively.
Also, each of the equations \eqref{eqns:CABC_ASBC} is assigned to a face of the cube.
Note that the face for Equation \eqref{eqn:CABC_S} corresponds to the cutting plane bisecting the cube diagonally, and the faces for the equations \eqref{eqn:CABC_B} and \eqref{eqn:CABC_BB} correspond to the upper triangles.
For simplicity, we here use the shorthand notations:~
$u=u_{l,m}$,\quad
$v=v_{l,m}$,\quad
$\overline{\rule{0em}{0.3em}\hspace{0.6em}}:\,l\to l+1$,\quad
$\widetilde{}\,:\,m\to m+1$.}
\label{fig:CABCcube}
\end{figure}

The CABC property is defined as follows.
\begin{enumerate}
\item 
There are following three ways to calculate $v_{l+1,m+1}$ by using all the equations in the system \eqref{eqns:CABC_ASBC} with $\{u_{l,m},u_{l+1,m},u_{l,m+1},v_{l,m}\}$ as initial values.
\begin{enumerate}
\item 
Express $v_{l+1,m+1}$ as a rational function in terms of $\{u_{l,m},u_{l+1,m},v_{l,m+1}\}$ by using Equation \eqref{eqn:CABC_S}.
Then, eliminate $v_{l,m+1}$ from it by using Equation \eqref{eqn:CABC_B}.\\[-0.7em]
\item 
Express $v_{l+1,m+1}$ as a rational function in terms of $\{u_{l+1,m},v_{l+1,m}\}$ by using Equation \eqref{eqn:CABC_BB}.
Then, eliminate $v_{l+1,m}$ from it by using Equation \eqref{eqn:CABC_C}.\\[-0.7em]
\item 
Express $v_{l+1,m+1}$ as a rational function in terms of $\{u_{l,m+1},u_{l+1,m+1},v_{l,m+1}\}$ by using Equation \eqref{eqn:CABC_CC}.
Then, eliminate $u_{l+1,m+1}$ by using Equation \eqref{eqn:CABC_A} and $v_{l,m+1}$ by using Equation \eqref{eqn:CABC_B} from it.\\[-0.7em]
\end{enumerate}
When $v_{l+1,m+1}$ is uniquely determined as a rational function with the initial values $\{u_{l,m},u_{l+1,m},u_{l,m+1},v_{l,m}\}$, 
then the system \eqref{eqns:CABC_PDE_ASBC} is said to have the {\it CABC property} or said to be a {\it CABC-system};
the {\PDE} \eqref{eqn:CABC_PDE_A} is said to have the CABC property. 
\item 
If there exists a {\PDE} given only by $\{v_{l,m},v_{l+1,m},v_{l,m+1},v_{l+1,m+1}\}$, then the tuple of equations \eqref{eqn:CABC_PDE_S}, \eqref{eqn:CABC_PDE_B} and \eqref{eqn:CABC_PDE_C} is referred to as a {\it CABC-type BT} from the {\PDE} \eqref{eqn:CABC_PDE_A} to the {\PDE} given by $\{v_{l,m},v_{l+1,m},v_{l,m+1},v_{l+1,m+1}\}$.
\end{enumerate}
\subsection{The CABC property of the dVolterra equation (I)}
The system of {\PDE}s:
\begin{subequations}
\begin{align}
 &(1+u_{l+1,m})u_{l,m}v_{l,m+1}v_{l+1,m+1}-\la v_{l+1,m+1}+1=0,\\
 &v_{l,m+1}=\dfrac{\la(1+u_{l,m})v_{l,m}-1}{u_{l,m}(v_{l,m}+u_{l,m}v_{l,m}+\la)},\\
 &(1+u_{l,m})u_{l+1,m}v_{l,m}v_{l+1,m}-\la v_{l+1,m}+1=0,
\end{align}
\end{subequations}
where $\la\in\bbC$,
is a CABC-type BT from the dVolterra equation \eqref{eqn:auto_dV} to the following multiquadratic quad-equation:
\begin{align}
 &(v_{l,m}v_{l+1,m}-v_{l,m+1}v_{l+1,m+1})^2-v_{l,m}v_{l+1,m}v_{l,m+1}v_{l+1,m+1}(v_{l,m}-v_{l+1,m+1})(v_{l+1,m}-v_{l,m+1})\notag\\
 &=\la(v_{l,m}v_{l+1,m}-v_{l,m+1}v_{l+1,m+1})\Big(v_{l,m}v_{l+1,m+1}(v_{l+1,m}-v_{l,m+1})\notag\\
 &\hspace{12.5em}+v_{l+1,m}v_{l,m+1}(v_{l,m}-v_{l+1,m+1})\Big).
\end{align}

Then, a Lax pair of the dVolterra equation \eqref{eqn:auto_dV} is given by \eqref{eqn:Lax_phi} where 
\begin{equation}
 L_{l,m}=
 \begin{pmatrix}
 0&1\\
 -(1+u_{l,m})u_{l+1,m}&\la
 \end{pmatrix},\quad
 M_{l,m}=
 \begin{pmatrix}
 \dfrac{\la(1+u_{l,m})}{u_{l,m}}&-\dfrac{1}{u_{l,m}}\\
 1+u_{l,m}&\la
 \end{pmatrix}.
\end{equation}
Indeed, the compatibility condition \eqref{eqn:compatibility_condition} gives the dVolterra equation \eqref{eqn:auto_dV}.
\subsection{The CABC property of the dVolterra equation (II)}
Here, we consider the dVolterra equation in the following form, obtained from \eqref{eqn:auto_dV} by setting $u_{k_1,k_2}\to u_{k_2,k_1}$:
\begin{equation}\label{eqn:auto_dV_2}
 \dfrac{~u_{l+1,m+1}~}{u_{l,m}}=\dfrac{1+u_{l,m+1}}{1+u_{l+1,m}}.
\end{equation}
The system of {\PDE}s:
\begin{subequations}
\begin{align}
&\dfrac{1+u_{l+1,m}}{u_{l,m}}=-\left(\dfrac{1+\la v_{l+1,m+1}}{\la-v_{l+1,m+1}}\right)\left(\dfrac{\la+v_{l,m+1}}{1-\la v_{l,m+1}}\right),\\
&u_{l,m}=-\dfrac{(\la-v_{l,m})(1-\la v_{l,m+1})}{2\la(1+v_{l,m}v_{l,m+1})},\\
&\dfrac{1+u_{l,m}}{u_{l+1,m}}=-\left(\dfrac{1+\la v_{l+1,m}}{\la-v_{l+1,m}}\right)\left(\dfrac{\la+v_{l,m}}{1-\la v_{l,m}}\right),
\end{align}
\end{subequations}
where $\la\in\bbC$,
 is a CABC-type BT from the dVolterra equation \eqref{eqn:auto_dV_2} to the following quad-equation:
\begin{align}
 &(v_{l,m}+v_{l+1,m})(1-v_{l,m+1}v_{l+1,m+1})-(v_{l,m+1}+v_{l+1,m+1})(1-v_{l,m}v_{l+1,m})\notag\\
 &=-\dfrac{1+\la^2}{\la}(v_{l,m}v_{l,m+1}-v_{l+1,m}v_{l+1,m+1}).
\end{align}

Then, a Lax pair of the dVolterra equation \eqref{eqn:auto_dV_2} is given by \eqref{eqn:Lax_phi} where
\begin{subequations}\label{eqns:Lax_LV}
\begin{align}
 &L_{l,m}=\dfrac{1}{{u_{l+1,m}}^{1/2}(1+u_{l,m})^{1/2}}
 \begin{pmatrix}
 \la^2+\la^2u_{l,m}-u_{l+1,m}&-\la(1+u_{l,m}+u_{l+1,m})\\
 \la(1+u_{l,m}+u_{l+1,m})&-1-u_{l,m}+\la^2u_{l+1,m}
 \end{pmatrix},\\[1em]
 &M_{l,m}=\dfrac{1}{{u_{l,m}}^{1/2}(1+u_{l,m})^{1/2}}
 \begin{pmatrix}
 -1&\la(1+2u_{l,m})\\
 -\la(1+2u_{l,m})&\la^2
 \end{pmatrix}.
\end{align}
\end{subequations}
Indeed, the compatibility condition \eqref{eqn:compatibility_condition} gives the dVolterra equation \eqref{eqn:auto_dV_2}.
\def\cprime{$'$} \def\cprime{$'$}


\begin{thebibliography}{10}

\bibitem{ABS2003:MR1962121}
V.~E. Adler, A.~I. Bobenko, and Y.~B. Suris.
\newblock Classification of integrable equations on quad-graphs. {T}he
  consistency approach.
\newblock {\em Comm. Math. Phys.}, 233(3):513--543, 2003.

\bibitem{ABS2009:MR2503862}
V.~E. Adler, A.~I. Bobenko, and Y.~B. Suris.
\newblock Discrete nonlinear hyperbolic equations: classification of integrable
  cases.
\newblock {\em Funktsional. Anal. i Prilozhen.}, 43(1):3--21, 2009.

\bibitem{bobenko1993discrete}
A.~Bobenko, N.~Kutz, and U.~Pinkall.
\newblock The discrete quantum pendulum.
\newblock {\em Physics Letters A}, 177(6):399--404, 1993.

\bibitem{BS2002:MR1890049}
A.~I. Bobenko and Y.~B. Suris.
\newblock Integrable systems on quad-graphs.
\newblock {\em Int. Math. Res. Not. IMRN}, (11):573--611, 2002.

\bibitem{BS2008:zbMATH05486618}
A.~I. Bobenko and Y.~B. Suris.
\newblock {\em Discrete differential geometry. {Integrable} structure},
  volume~98 of {\em Grad. Stud. Math.}
\newblock Providence, RI: American Mathematical Society (AMS), 2008.

\bibitem{BollR2011:MR2846098}
R.~Boll.
\newblock Classification of 3{D} consistent quad-equations.
\newblock {\em J. Nonlinear Math. Phys.}, 18(3):337--365, 2011.

\bibitem{BollR:thesis}
R.~Boll.
\newblock Classification and {L}agrangian {S}tructure of 3{D} {C}onsistent
  {Q}uad-{E}quations.
\newblock {\em Doctoral Thesis, Technische Universit\"at Berlin}, 2012.

\bibitem{BollR2012:MR3010833}
R.~Boll.
\newblock Corrigendum: {C}lassification of 3{D} consistent quad-equations.
\newblock {\em J. Nonlinear Math. Phys.}, 19(4):1292001, 3, 2012.

\bibitem{FJN2008:MR2425981}
C.~M. Field, N.~Joshi, and F.~W. Nijhoff.
\newblock {$q$}-difference equations of {K}d{V} type and {C}hazy-type
  second-degree difference equations.
\newblock {\em J. Phys. A}, 41(33):332005, 13, 2008.

\bibitem{FuchsR1905:quelques}
R.~Fuchs.
\newblock Sur quelques \'equations diff\'erentielles lin\'eaires du second
  ordre.
\newblock {\em Comptes Rendus de l'Acad\'emie des Sciences Paris},
  141(1):555--558, 1905.

\bibitem{GambierB1910:MR1555055}
B.~Gambier.
\newblock Sur les \'equations diff\'erentielles du second ordre et du premier
  degr\'e dont l'int\'egrale g\'en\'erale est a points critiques fixes.
\newblock {\em Acta Math.}, 33(1):1--55, 1910.

\bibitem{GarnierR1912:MR1509146}
R.~Garnier.
\newblock Sur des \'equations diff\'erentielles du troisi\`eme ordre dont
  l'int\'egrale g\'en\'erale est uniforme et sur une classe d'\'equations
  nouvelles d'ordre sup\'erieur dont l'int\'egrale g\'en\'erale a ses points
  critiques fixes.
\newblock {\em Ann. Sci. \'Ecole Norm. Sup. (3)}, 29:1--126, 1912.

\bibitem{GRSWC2005:MR2117991}
B.~Grammaticos, A.~Ramani, J.~Satsuma, R.~Willox, and A.~S. Carstea.
\newblock Reductions of integrable lattices.
\newblock {\em J. Nonlinear Math. Phys.}, 12(suppl. 1):363--371, 2005.

\bibitem{HHJN2007:MR2303490}
M.~Hay, J.~Hietarinta, N.~Joshi, and F.~W. Nijhoff.
\newblock A {L}ax pair for a lattice modified {K}d{V} equation, reductions to
  {$q$}-{P}ainlev\'e equations and associated {L}ax pairs.
\newblock {\em J. Phys. A}, 40(2):F61--F73, 2007.

\bibitem{HHNS2015:MR3317164}
M.~Hay, P.~Howes, N.~Nakazono, and Y.~Shi.
\newblock A systematic approach to reductions of type-{Q} {ABS} equations.
\newblock {\em J. Phys. A}, 48(9):095201, 24, 2015.

\bibitem{HietarintaJ2019:zbMATH07053246}
J.~Hietarinta.
\newblock Search for {CAC}-integrable homogeneous quadratic triplets of quad
  equations and their classification by {BT} and {Lax}.
\newblock {\em J. Nonlinear Math. Phys.}, 26(3):358--389, 2019.

\bibitem{HJN2016:MR3587455}
J.~Hietarinta, N.~Joshi, and F.~W. Nijhoff.
\newblock {\em Discrete systems and integrability}.
\newblock Cambridge Texts in Applied Mathematics. Cambridge University Press,
  Cambridge, 2016.

\bibitem{HV2007:Searching}
J.~Hietarinta and C.~Viallet.
\newblock Searching for integrable lattice maps using factorization.
\newblock {\em Journal of Physics A: Mathematical and Theoretical},
  40(42):12629--12643, 2007.

\bibitem{HirotaR1977:MR0460934}
R.~Hirota.
\newblock Nonlinear partial difference equations. {I}. {A} difference analogue
  of the {K}orteweg-de {V}ries equation.
\newblock {\em J. Phys. Soc. Japan}, 43(4):1424--1433, 1977.

\bibitem{hirota1976n}
R.~Hirota and J.~Satsuma.
\newblock N-soliton solutions of nonlinear network equations describing a
  volterra system.
\newblock {\em Journal of the Physical Society of Japan}, 40(3):891--900, 1976.

\bibitem{hirota1995conserved}
R.~Hirota and S.~Tsujimoto.
\newblock Conserved quantities of a class of nonlinear difference-difference
  equations.
\newblock {\em Journal of the Physical Society of Japan}, 64(9):3125--3127,
  1995.

\bibitem{JKMN2021:zbMATH07653201}
N.~Joshi, K.~Kajiwara, T.~Masuda, and N.~Nakazono.
\newblock Discrete power functions on a hexagonal lattice. {I}: {Derivation} of
  defining equations from the symmetry of the {Garnier} system in two
  variables.
\newblock {\em J. Phys. A, Math. Theor.}, 54(33):27, 2021.
\newblock Id/No 335202.

\bibitem{JKMNS2017:MR3741826}
N.~Joshi, K.~Kajiwara, T.~Masuda, N.~Nakazono, and Y.~Shi.
\newblock Geometric description of a discrete power function associated with
  the sixth {P}ainlev\'{e} equation.
\newblock {\em Proc. R. Soc. A.}, 473(2207):20170312, 19, 2017.

\bibitem{JN2021:zbMATH07476241}
N.~Joshi and N.~Nakazono.
\newblock On the three-dimensional consistency of {Hirota}'s discrete
  {Korteweg}-de {Vries} equation.
\newblock {\em Stud. Appl. Math.}, 147(4):1409--1424, 2021.

\bibitem{KM1975:MR0369953}
M.~Kac and P.~van Moerbeke.
\newblock On an explicitly soluble system of nonlinear differential equations
  related to certain {T}oda lattices.
\newblock {\em Advances in Math.}, 16:160--169, 1975.

\bibitem{KNY2017:MR3609039}
K.~Kajiwara, M.~Noumi, and Y.~Yamada.
\newblock Geometric aspects of {P}ainlev\'{e} equations.
\newblock {\em J. Phys. A}, 50(7):073001, 164, 2017.

\bibitem{KimuraH1990:MR1078402}
H.~Kimura.
\newblock Symmetries of the {G}arnier system and of the associated polynomial
  {H}amiltonian system.
\newblock {\em Proc. Japan Acad. Ser. A Math. Sci.}, 66(7):176--178, 1990.

\bibitem{KO1984:MR776915}
H.~Kimura and K.~Okamoto.
\newblock On the polynomial {H}amiltonian structure of the {G}arnier systems.
\newblock {\em J. Math. Pures Appl. (9)}, 63(1):129--146, 1984.

\bibitem{KDV1895:zbMATH02679684}
D.~J. Korteweg and G.~De~Vries.
\newblock On the change of form of long waves advancing in a rectangular canal,
  and on a new type of long stationary waves.
\newblock {\em Phil. Mag. (5)}, 39:422--443, 1895.

\bibitem{NakazonoN2018:MR3760161}
N.~Nakazono.
\newblock Reduction of lattice equations to the {P}ainlev\'e equations: {$\rm
  P_{IV}$} and {$\rm P_V$}.
\newblock {\em J. Math. Phys.}, 59(2):022702, 18, 2018.

\bibitem{nakazono2022discrete}
N.~Nakazono.
\newblock Discrete {P}ainlev\'e transcendent solutions to the
  multiplicative-type discrete {K}d{V} equations.
\newblock {\em Journal of Mathematical Physics}, 63(4):042703, 2022.

\bibitem{nakazono2024consistency}
N.~Nakazono.
\newblock Consistency around a cube property of {H}irota's discrete {K}d{V}
  equation and the lattice sine-{G}ordon equation.
\newblock {\em Applied Numerical Mathematics}, 199:136--152, 2024.

\bibitem{nakazono2026variationdKdV}
N.~Nakazono.
\newblock Solutions to an autonomous discrete kdv equation via painlev\'e-type.
\newblock {\em arXiv preprint arXiv:2503.06013}, 2025.

\bibitem{NijhoffFW2002:MR1912127}
F.~W. Nijhoff.
\newblock Lax pair for the {A}dler (lattice {K}richever-{N}ovikov) system.
\newblock {\em Phys. Lett. A}, 297(1-2):49--58, 2002.

\bibitem{NC1995:MR1329559}
F.~W. Nijhoff and H.~W. Capel.
\newblock The discrete {K}orteweg-de {V}ries equation.
\newblock {\em Acta Appl. Math.}, 39(1-3):133--158, 1995.
\newblock KdV '95 (Amsterdam, 1995).

\bibitem{NCWQ1984:MR763123}
F.~W. Nijhoff, H.~W. Capel, G.~L. Wiersma, and G.~R.~W. Quispel.
\newblock B\"acklund transformations and three-dimensional lattice equations.
\newblock {\em Phys. Lett. A}, 105(6):267--272, 1984.

\bibitem{NP1991:MR1098879}
F.~W. Nijhoff and V.~G. Papageorgiou.
\newblock Similarity reductions of integrable lattices and discrete analogues
  of the {P}ainlev\'e {${\rm II}$} equation.
\newblock {\em Phys. Lett. A}, 153(6-7):337--344, 1991.

\bibitem{NQC1983:MR719638}
F.~W. Nijhoff, G.~R.~W. Quispel, and H.~W. Capel.
\newblock Direct linearization of nonlinear difference-difference equations.
\newblock {\em Phys. Lett. A}, 97(4):125--128, 1983.

\bibitem{NRGO2001:MR1819383}
F.~W. Nijhoff, A.~Ramani, B.~Grammaticos, and Y.~Ohta.
\newblock On discrete {P}ainlev\'e equations associated with the lattice
  {K}d{V} systems and the {P}ainlev\'e {VI} equation.
\newblock {\em Stud. Appl. Math.}, 106(3):261--314, 2001.

\bibitem{NW2001:MR1869690}
F.~W. Nijhoff and A.~J. Walker.
\newblock The discrete and continuous {P}ainlev\'e {VI} hierarchy and the
  {G}arnier systems.
\newblock {\em Glasg. Math. J.}, 43A:109--123, 2001.
\newblock Integrable systems: linear and nonlinear dynamics (Islay, 1999).

\bibitem{NS1998:zbMATH01844203}
J.~J.~C. Nimmo and W.~K. Schief.
\newblock An integrable discretization of a {{\((2+1)\)}}-dimensional
  sine-{Gordon} equation.
\newblock {\em Stud. Appl. Math.}, 100(3):295--309, 1998.

\bibitem{book_NoumiM2004:MR2044201}
M.~Noumi.
\newblock {\em Painlev\'e equations through symmetry}, volume 223 of {\em
  Translations of Mathematical Monographs}.
\newblock American Mathematical Society, Providence, RI, 2004.
\newblock Translated from the 2000 Japanese original by the author.

\bibitem{OKSO2006:MR2277519}
Y.~Ohyama, H.~Kawamuko, H.~Sakai, and K.~Okamoto.
\newblock Studies on the {P}ainlev\'e equations. {V}. {T}hird {P}ainlev\'e
  equations of special type {$P_{\rm III}(D_7)$} and {$P_{\rm III}(D_8)$}.
\newblock {\em J. Math. Sci. Univ. Tokyo}, 13(2):145--204, 2006.

\bibitem{OkamotoK1986:MR854008}
K.~Okamoto.
\newblock Studies on the {P}ainlev\'e equations. {III}. {S}econd and fourth
  {P}ainlev\'e equations, {$P_{{\rm II}}$} and {$P_{{\rm IV}}$}.
\newblock {\em Math. Ann.}, 275(2):221--255, 1986.

\bibitem{OkamotoK1987:MR916698}
K.~Okamoto.
\newblock Studies on the {P}ainlev\'e equations. {I}. {S}ixth {P}ainlev\'e
  equation {$P_{{\rm VI}}$}.
\newblock {\em Ann. Mat. Pura Appl. (4)}, 146:337--381, 1987.

\bibitem{OkamotoK1987:MR914314}
K.~Okamoto.
\newblock Studies on the {P}ainlev\'e equations. {II}. {F}ifth {P}ainlev\'e
  equation {$P_{\rm V}$}.
\newblock {\em Japan. J. Math. (N.S.)}, 13(1):47--76, 1987.

\bibitem{OkamotoK1987:MR927186}
K.~Okamoto.
\newblock Studies on the {P}ainlev\'e equations. {IV}. {T}hird {P}ainlev\'e
  equation {$P_{{\rm III}}$}.
\newblock {\em Funkcial. Ekvac.}, 30(2-3):305--332, 1987.

\bibitem{OK1986:zbMATH03962350}
K.~Okamoto and H.~Kimura.
\newblock On particular solutions of the {Garnier} systems and the
  hypergeometric functions of several variables.
\newblock {\em Q. J. Math., Oxf. II. Ser.}, 37:61--80, 1986.

\bibitem{OrmerodCM2012:MR2997166}
C.~M. Ormerod.
\newblock Reductions of lattice m{K}d{V} to {$q$}-{${\rm P}_{\rm VI}$}.
\newblock {\em Phys. Lett. A}, 376(45):2855--2859, 2012.

\bibitem{OrmerodCM2014:MR3210633}
C.~M. Ormerod.
\newblock Symmetries and special solutions of reductions of the lattice
  potential {K}d{V} equation.
\newblock {\em SIGMA Symmetry Integrability Geom. Methods Appl.}, 10:Paper 002,
  19, 2014.

\bibitem{PainleveP1900:zbMATH02665472}
P.~Painlev{\'e}.
\newblock M{\'e}moire sur les {\'e}quations diff{\'e}rentielles dont
  l'int{\'e}grale g{\'e}n{\'e}rale est uniforme.
\newblock {\em Bull. Soc. Math. Fr.}, 28:201--261, 1900.

\bibitem{PainleveP1902:MR1554937}
P.~Painlev{\'e}.
\newblock Sur les \'equations diff\'erentielles du second ordre et d'ordre
  sup\'erieur dont l'int\'egrale g\'en\'erale est uniforme.
\newblock {\em Acta Math.}, 25(1):1--85, 1902.

\bibitem{PainleveP1907:zbMATH02647172}
P.~Painlev{\'e}.
\newblock Sur les {\'e}quations diff{\'e}rentielles du second ordre {\`a}
  points critiques fixes.
\newblock {\em C. R. Acad. Sci., Paris}, 143:1111--1117, 1907.

\bibitem{QNCL1984:MR761644}
G.~R.~W. Quispel, F.~W. Nijhoff, H.~W. Capel, and J.~van~der Linden.
\newblock Linear integral equations and nonlinear difference-difference
  equations.
\newblock {\em Phys. A}, 125(2-3):344--380, 1984.

\bibitem{SakaiH2001:MR1882403}
H.~Sakai.
\newblock Rational surfaces associated with affine root systems and geometry of
  the {P}ainlev\'e equations.
\newblock {\em Comm. Math. Phys.}, 220(1):165--229, 2001.

\bibitem{SuzukiT2005:MR2177118}
T.~Suzuki.
\newblock Affine {W}eyl group symmetry of the {G}arnier system.
\newblock {\em Funkcial. Ekvac.}, 48(2):203--230, 2005.

\bibitem{TakenawaT2024:RIMS}
T.~Takenawa.
\newblock {S}pace of initial conditions for the four-dimensional {G}arnier
  system revisited ({R}ecent developments in mathematics of integrable
  systems).
\newblock {\em RIMS K\^oky\^uroku Bessatsu}, B96:117--130, 2024.

\bibitem{TX2009:zbMATH05559172}
D.~Tsoubelis and P.~Xenitidis.
\newblock Continuous symmetric reductions of the {Adler}-{Bobenko}-{Suris}
  equations.
\newblock {\em J. Phys. A, Math. Theor.}, 42(16):29, 2009.
\newblock Id/No 165203.

\bibitem{TsudaT2003:MR1987136}
T.~Tsuda.
\newblock Birational symmetries, {H}irota bilinear forms and special solutions
  of the {G}arnier systems in 2-variables.
\newblock {\em J. Math. Sci. Univ. Tokyo}, 10(2):355--371, 2003.

\bibitem{TsudaT2003:MR1998470}
T.~Tsuda.
\newblock Rational solutions of the {G}arnier system in terms of {S}chur
  polynomials.
\newblock {\em Int. Math. Res. Not. IMRN}, (43):2341--2358, 2003.

\bibitem{TsudaT2006:MR2263717}
T.~Tsuda.
\newblock Toda equation and special polynomials associated with the {G}arnier
  system.
\newblock {\em Adv. Math.}, 206(2):657--683, 2006.

\bibitem{volkov1992quantum}
A.~Y. Volkov and L.~D. Faddeev.
\newblock Quantum inverse scattering method on a spacetime lattice.
\newblock {\em Theoretical and Mathematical Physics}, 92(2):837--842, 1992.

\bibitem{WalkerAJ:thesis}
A.~Walker.
\newblock Similarity reductions and integrable lattice equations.
\newblock {\em Ph.D. Thesis, University of Leeds}, 2001.

\bibitem{YamadaY2009:zbMATH05588071}
Y.~Yamada.
\newblock Pad{\'e} method to {Painlev{\'e}} equations.
\newblock {\em Funkc. Ekvacioj, Ser. Int.}, 52(1):83--92, 2009.

\end{thebibliography}
\end{document}